\documentclass[12pt]{article}
\usepackage{amsmath,amssymb,amsthm}
\usepackage{color}
\usepackage{graphicx}
\usepackage{tikz}
\usepackage{enumerate}

\theoremstyle{plain}
\newtheorem{theorem}{Theorem}
\newtheorem{corollary}{Corollary}
\newtheorem{lemma}{Lemma}
\newtheorem{proposition}{Proposition}

\theoremstyle{definition}
\newtheorem{definition}{Definition}

\begin{document}

\def\nonext{\hbox{*}}
\def\set1n{\{1,\ldots,n\}}
\def\Naturals{{\bf N}}
\def\Integers{{\bf Z}}
\def\starN{\nonext\Naturals}
\def\Rationals{{\bf Q}}
\def\Reals{{\bf R}}
\def\starR{\nonext\Reals}
\def\starRplus{\nonext \Reals_+}
\def\Rplus{\Reals_+}
\def\Rplusplus{{\Reals_{++}}}
\def\Rk{\Reals^k}
\def\Rdplus{\Reals^d_+}
\def\Rd{\Reals^d}
\def\Rm{\Reals^m}
\def\RJ{\Reals^J}
\def\RJplus{\Reals^J_+}
\def\Rmplus{\Reals^m_+}
\def\RK{\Reals^K}
\def\Rkplus{{\Reals^k_+}}
\def\Rkplusplus{{\Reals^k_{++}}}
\def\Rtwoplus{{\Reals^2_+}}
\def\Rtwoplusplus{{\Reals^2_{++}}}
\def\Deltaplusplus{{\Delta_{++}}}
\def\Deltaplus{{\Delta_{+}}}
\def\starpref{{\nonext \hspace{-.05in} \succ}}
\def\notstarpref{{\nonext \hspace{-.05in} \not\succ}}
\def\stpref{{\st \hspace{-.05in} \succ}}
\def\st{{{}^\circ}}
\def\stin{{\rm st}^{-1}}
\def\Var{{\rm Var}}
\def\ns{{\rm ns}}

\setlength{\parskip}{.15in}

\title{Uncertainty and Robustness of Surplus Extraction\thanks{We thank the Co-Editor Marciano Siniscalchi and referees for helpful comments. Financial support from the NSF under grants SES 0721145 and SES 1227707 is gratefully acknowledged. Some of the results in this paper originally appeared in the working paper ``Uncertainty in Mechanism Design,'' Lopomo, Rigotti, and Shannon (2009). Contact: giuseppe.lopomo@duke.edu,  luca@pitt.edu, cshannon@econ.berkeley.edu }}
\date{May 2020}
\author{Giuseppe Lopomo \\ Duke University \and Luca Rigotti\\ University of Pittsburgh \and Chris Shannon \\ UC Berkeley}
\maketitle
\begin{abstract}
This paper studies a robust version of the classic surplus extraction problem, in which the designer knows only that the beliefs of each type belong to some set, and designs mechanisms that are suitable for all possible beliefs in that set. We derive necessary and sufficient conditions for full extraction in this setting, and show that these are natural set-valued analogues of the classic convex independence condition identified by Cr{\' e}mer and McLean (1985, 1988). We show that full extraction is neither generically possible nor generically impossible, in contrast to the standard setting in which full extraction is generic. When full extraction fails, we show that natural additional conditions can restrict both the nature of the contracts a designer can offer and the surplus the designer can obtain. 
\end{abstract}

\setlength{\parskip}{.15in}

\section{Introduction}

This paper studies a robust version of the classic surplus extraction problem. In the standard setting, when agents' private information is correlated with their beliefs, a designer can typically extract all or virtually all information rents in a broad range of environments, as shown in the pioneering work of Cr{\' e}mer and McLean (1985, 1988) and McAfee and Reny (1992). We consider a robust version of this problem, in which the designer knows only that the beliefs of each type belong to some set, and designs mechanisms that are suitable for all possible beliefs in that set. We derive necessary and sufficient conditions for full extraction in this setting, and show that these are natural set-valued analogues of the classic convex independence condition identified by Cr{\' e}mer and McLean (1985, 1988). We show that full extraction is neither generically possible nor generically impossible, in contrast to the standard setting in which full extraction is generic. When full extraction fails, we show that natural additional conditions can restrict both the nature of the contracts a designer can offer and the surplus the designer can obtain.

The classic full extraction results of Cr{\' e}mer and McLean (1985, 1988) and McAfee and Reny (1992) create a puzzle for mechanism design because they suggest that
private information often has no value to agents. In addition, these results suggest that complicated stochastic mechanisms are typically optimal for the designer. Both predictions seem unrealistic. In the literature that followed, a number of papers showed that full extraction breaks down under natural modifications to the setting. In particular, risk aversion, limited
liability, budget constraints, collusion, and competition are among the reasons that
full extraction may fail (see Robert (1991), Laffont and Martimort (2000), Peters (2001), or Che and Kim (2006), for example). On the other hand, relaxing the standard setting to include these features often eliminates other desirable properties, while also typically predicting complex mechanisms that bear little resemblance to those frequently used in practice. In this paper, we show not only that our robustness requirements can limit the designer's ability to fully extract rents, but also that when this happens the designer can be limited to simpler and more realistic mechanisms.

Under the notion of robustness in this model, whether or not full extraction is possible depends in part on the amount of uncertainty about agents' beliefs. If agents' beliefs are given precisely, the model reduces to the standard problem, and full extraction is possible generically. In particular, in this case full extraction holds if agents' beliefs satisfy the convex independence condition of Cr{\' e}mer and McLean (1988) (which is satisfied on a generic set). Starting from such beliefs, we show that full extraction continues to hold with uncertainty about these beliefs, provided uncertainty is sufficiently small. Full extraction fails when uncertainty is sufficiently large, however. We show  that full extraction holds in general if beliefs satisfy a natural set-valued version of convex independence. When this condition is satisfied, the mechanisms that achieve full extraction are variants of the full extraction mechanisms constructed by Cr{\' e}mer and McLean (1988). When full extraction fails, however, we show that the designer can be restricted to simple mechanisms: natural conditions limit the designer to offering a single contract, and can make a deterministic contract optimal for the designer. 

We follow McAfee and Reny (1992) in considering a reduced form description of the surplus extraction problem. In a prior, unmodeled stage, agents play a game that leaves them with some information rents as a function of their private information. This game could be an auction, a bargaining game, or a purchase from a seller, for example. Private information is summarized by the agent's type.  The current stage also has an exogenous source of uncertainty, summarized by a set of states, on which contract payments can depend. For some applications, it is natural to take the state space to be the set of types, although we follow McAfee and Reny (1992) in allowing the state space to be arbitrary. We assume that both the set of types and the state space are finite, as in the original environments considered by Cr{\' e}mer and McLean (1985, 1988).\footnote{McAfee and Reny (1992) consider an infinite set of types, and allow for an infinite state space; their results characterize virtual extraction as a consequence.} Each type determines both the information rent and a set of beliefs over the state space. 

Our main results define and characterize full extraction in this model. We start with notions of full extraction motivated by robustness to uncertainty about agents' beliefs; these notions require incentive compatibility and individual rationality conditions to hold for all possible beliefs of agents. We then connect these notions to choice behavior under Knightian uncertainty, as in the foundational model of Bewley (1986) (see also Bewley (2002)). We show that these notions of extraction are closely related, and that the same set-valued analogues of convex independence provide necessary and sufficient conditions for full extraction in each case. Our results can thus be interpreted as incorporating robustness to the designer's misspecification of agents' beliefs, or robustness to agents' perceptions of uncertainty. We show that full extraction is neither generically possible nor generically impossible in this setting, and that this result holds under both topological and measure-theoretic notions of genericity. This is in contrast with the standard environment in which each type is associated with a single belief, either because the designer is not concerned about misspecification or because agents do not perceive uncertainty, and convex independence is satisfied for a generic subset of type-dependent beliefs whenever there are at least as many states as types. 

We also explore limits on the complexity of contracts the designer can offer. When convex independence is satisfied and full extraction holds, the designer typically offers a menu of contracts with as many different contracts as types. When convex independence fails, however, the designer can be restricted both in the complexity of contracts offered and in the surplus that can be extracted. In the spirit of our earlier work (Lopomo et al. (2009)), we show that this is particularly salient when the beliefs associated with different types intersect. We adapt the notion of fully overlapping beliefs introduced in Lopomo et al. (2009), which gives a notion of richness of the beliefs common to several types, to this setting with finitely many types. When beliefs are fully overlapping, we show that the designer is limited to offering a single contract. Under additional conditions on the designer's beliefs and objective, we show that a single deterministic contract can be optimal for the designer.      

Our paper is related to several strands of literature on robustness and ambiguity in mechanism design problems more generally, and on surplus extraction more specifically. We adapt the basic framework of our earlier work in Lopomo et al. (2009) to the general setting of McAfee and Reny (1992), and restrict attention to problems with a finite set of states and types. Our earlier work instead considers more general mechanism design problems in settings with an infinite set of types, and gives conditions under which incentive compatible mechanisms must be simple, in particular ex-post incentive compatible. The key condition we introduce in Lopomo et al. (2009) is a version of fully overlapping beliefs. We illustrate using the leading example of epsilon-contamination, which we adapt here in section 3. As in the current paper, we argue that the results of Lopomo et al. (2009) can be interpreted either as robustness or ambiguity. Jehiel et al.  (2012) adopt the model of robustness introduced in Lopomo et al. (2009), and show that ex-post implementation is generically impossible in the epsilon-contamination example. Chiesa et al. (2015) also adapt the model of Lopomo, Rigotti, and Shannon (2009) to a setting with a finite set of types, and focus on the performance of Vickrey mechanisms in either dominant or undominated strategies. Fu et al. (2017) also  consider the problem of surplus extraction with finitely many types when the designer does not have full information about the distribution of buyers' valuations. They consider auctions in which the seller can observe samples from a finite set of possible distributions of buyers' valuations, and can condition the mechanism on these samples. They give tight bounds on the number of samples needed for full extraction using dominant strategy incentive compatible mechanisms. Albert et al. (2019) study related sampling and estimation problems for a monopolistic seller facing a single buyer, and consider mechanisms that relax incentive and participation constraints by only requiring that these hold with high probability. Their focus is  primarily on computational aspects of finding optimal mechanisms in this class. They give conditions under which the optimal such mechanism can be computed in polynomial time, and under which it can be learned using a polynomial number of samples from the buyer's distribution.  Instead, we consider the more general setting of McAfee and Reny (1992), and assume that the designer does not have access to samples from the distribution of buyers' rents, while knowing the set of possible such distributions. We then give conditions under which full surplus extraction is possible regardless of buyers' beliefs.   

Our paper also connects with a growing body of work on mechanism design under ambiguity. This includes Bodoh-Creed (2012), Bose et al. (2006), and Wolitzky (2016), which extend standard mechanism design problems to maxmin expected utility agents, and Bose and Renou (2014) and De Tillio et al. (2016), which allow for maxmin expected utility agents facing ambiguity about aspects of the mechanism. In particular, Wolitzky (2016) derives a necessary condition for social choice rules to be implementable with maxmin expected utility agents, and provides a characterization of when efficient trade is implementable in bilateral trade, extending the classic Myerson-Satterthwaite theorem (Myerson and Satterthwaite (1983)). Most closely related is Bose and Daripa (2009), who study auction design and surplus extraction in a model with ambiguity. They consider an independent private values setting in which bidders are assumed to have maxmin expected utility with beliefs modeled using epsilon-contamination. They show that the seller can extract almost all of the surplus by using a dynamic mechanism which is a modified Dutch auction. This is in contrast to results for optimal auctions without ambiguity, or to optimal static auctions with ambiguity, which leave rents to all but the lowest types (see Bose, Ozdenoren and Pape (2006)).  We focus instead on an analogue of the correlated beliefs setting of Cr{\' e}mer and McLean (1988), for which full extraction holds generically under unique priors using static mechanisms. Our results show that full extraction can still hold with ambiguity, using variants of the static mechanisms of Cr{\' e}mer and McLean (1988), but that full extraction is no longer a generic feature of the model with ambiguity.   

Our paper is also related to a substantial literature exploring the limits of the classic full extraction results. As discussed above, risk aversion, limited liability, budget constraints, collusion, and competition have all been shown to be reasons full extraction might fail. Our results can be understood as showing that robustness to sufficiently imprecise information about beliefs or sufficient ambiguity might be other reasons for the failure of full extraction. Another important recent strand of work questions whether the conclusion that full extraction holds generically is robust to alternative models of agents' beliefs and higher order beliefs, to the relationship between payoffs and beliefs, or to the notion of genericity used. Neeman (2004) focuses on the relationship between payoff and beliefs, and notes that full extraction fails when different rents can be associated with the same beliefs. Heifetz and Neeman (2006) show that the type spaces in which this is ruled out, and thus ``beliefs determine preferences'' are required and full extraction is possible, are not generic in a measure-theoretic sense within the universal type space. Barelli (2009) and Chen and Xiong (2011) argue that whether such type spaces are generic or not depends on the notion of genericity used, and show that such type spaces are instead topologically generic. Chen and Xiong (2013) show  that full extraction also  holds generically, again in a topological sense. Chen and Xiong (2013) establish their generic result by showing that full extraction mechanisms are robust to sufficiently small changes in priors, in that for each $\varepsilon >0$, if a mechanism from a particular class extracts all but at most $\varepsilon$ surplus for a given prior, then the same mechanism also extracts all but $\varepsilon$ surplus for all priors in a sufficiently small weak-$^*$ neighborhood of the original prior. Our paper complements these results by showing that full extraction can hold robustly in a related sense. In our setting, full extraction is neither generically possible nor generically impossible, regardless of whether topological or measure-theoretic notions of genericity are invoked.\footnote{Ahn (2007) establishes the existence of an analogue of the Mertens-Zamir universal type space for
hierarchies of ambiguous sets of beliefs. As he shows, ambiguity at any level of the belief hierarchy corresponds to a set of beliefs over types.} 

The paper proceeds as follows. In section 2, we set up the basic model and definitions, and discuss the model of incomplete preferences. In section 3, we give the leading example to illustrate the main ideas and results of the paper. In section 4 we develop the general model, and give our main results characterizing sufficient and necessary conditions for full extraction. We also connect these results to choice behavior under Knightian uncertainty, and provide alternative results for notions of full extraction motivated by such behavior. We also consider limits to the designer's ability to extract rents when these results do not apply. In section 5 we consider the genericity of full extraction. Finally, in section 6 we consider several extensions of our main results. We explore the performance of given extraction mechanisms when beliefs are perturbed, and show that several of our main results carry over to other standard models of ambiguity, including versions of maxmin and alpha-maxmin expected utility.

\section{Set-up and Preliminaries}

In this section, we give the  set up of the model and questions we address. We next recall some preliminary definitions and define some basic notation that we will use throughout. Finally, we discuss the  model of incomplete preferences, motivated by Bewley (1986).

\subsection{Set-up and Extraction Notions}

We first lay out the basic set-up, definitions, and notation used throughout the paper, and give the definitions of surplus extraction underlying the main results.

We follow McAfee and Reny (1992) in giving a reduced form description of the surplus extraction problem. In a prior, unmodeled stage, agents play a game that leaves them with some rents as a function of their private information. The game could be an auction, a bargaining game, or a purchase from a seller, for example. The designer (a seller, a mediator, etc.) can charge the agent for participating in the game, while the agent can choose whether or not to participate. If the agent does not participate, her payoff is zero. 

Private information is summarized by the type $t\in T$, where $T$ denotes the set of possible types.  The current stage also has an exogenous source of uncertainty, summarized by a set of states $S$, on which contract payments can depend. For some applications, it is natural to take $S=T$, although we follow McAfee and Reny (1992) in allowing $S$ to be arbitrary. We assume throughout that $S$ and $T$ are finite, and use $S$ and $T$ to denote both the sets and their cardinalities. We use the standard notation $\Delta(A)$ to denote the set of probabilities on a finite set $A$; in particular, $\Delta(S) = \{ \pi\in{\bf R}^S_+: \sum_{s\in S} \pi_s = 1\}$.

Each type $t\in T$ is then associated with a value $v(t)\in {\bf R}$, representing the rents from the prior stage, and a set of beliefs $\Pi(t) \subseteq \Delta (S)$, which we assume is closed and convex.  The information rent can also be allowed to depend on the public information $S$, with additional notation and steps in several places, but no change in any of the results. 

The designer then offers a menu of stochastic contracts $\{ c(t): t\in T\} \subseteq {\bf R}^S$, from which agents select a contract. After the state is realized, agents pay the designer the amount specified by the contract in that state.

In the classic setting of Cr{\' e}mer and McLean (1985, 1988) and McAfee and Reny (1992), each type $t\in T$ is associated with a unique belief $\pi(t)\in \Delta(S)$, and typically {\it full extraction} holds, that is, given any $v:T\to {\bf R}$, there exists a menu of contracts $\{c(t) \in {\bf R}^S : t\in T\}$ such that for each $t\in T$:
\[ 
v(t) - \pi(t)\cdot c(t) = 0
\]
and
\[
v(t) -\pi(t)\cdot c(s) \leq 0 \ \ \ \forall s\not= t
\]

We consider two natural extensions of this notion to account for uncertainty in agents' beliefs. Both reflect the idea that the designer offers agents a menu of stochastic contracts from which they choose, based on minimizing their expected costs. The first requires contracts to be uniformly ranked by all agents for all beliefs, and thus strengthens standard incentive constraints in response to uncertainty. We take this to be the robust version of full extraction in this setting. The second is weaker, instead requiring that contracts are ranked only for some beliefs. This notion weakens incentive constraints in response to uncertainty.

\begin{definition}
\emph{Full extraction} holds if, given $v:T\to {\bf R}$, there exists a menu of contracts $\{ c(t) \in {\bf R}^S: t\in T\}$ such that for each $t\in T$, 
\[
v(t) -\pi(t) \cdot c(t) \geq 0 \ \ \forall \pi\in \Pi(t), \mbox{ with } v(t) = \pi(t) \cdot c(t) \ \mbox{ for some } \pi\in \Pi(t)
\]
and
\[
v(t) - \pi(t) \cdot c(s) \leq 0 \ \ \forall \pi\in \Pi(t), \ \ \forall s\not= t
\]
\end{definition}

\bigskip

\begin{definition}
\emph{Weak full extraction} holds if, given $v:T\to {\bf R}$, there exists a menu of contracts $\{ c(t) \in {\bf R}^S: t\in T\}$ such that for each $t\in T$,
\[
v(t) -\pi(t) \cdot c(t) = 0 
\]
and
\[
v(t) - \pi(t) \cdot c(s) \leq 0 \ \ \forall s\not= t
\]
for some $\pi(t) \in \Pi(t)$. 
\end{definition}

Note that in either case extraction requires that for each type $t\in T$, the expected value of the contract $c(t)$ is equal to his value $v(t)$ for (at least) one belief $\pi(t) \in \Pi(t)$. While it might be natural to define extraction using the stronger requirement that $v(t) = \pi\cdot c(t)$ for all $\pi\in \Pi(t)$, this will rule out stochastic contracts for sufficiently rich beliefs, and hence extraction will typically not be possible under this definition. This point is related to other results on the limits to full extraction that we develop at the end of section 4.\footnote{In particular, if $\Pi(t)$ has full dimension (defined at the end of this section), then $\pi\cdot c(t) = v(t)$ for all $\pi\in \Pi(t)$ if and only if $c(t) = (v(t), v(t),\ldots ,v(t))$, that is, if and only if $c(t)$ is the constant contract with value $v(t)$ in each state. Full extraction is then impossible for types $s\not= t$ with $v(s) > v(t)$.}

For now we take these as primitive definitions. We show below in section 4 that these notions can be naturally related to agents' choice behavior when  the model is interpreted using ambiguity, and also discuss connections to several other natural notions. 

Following Cr{\' e}mer and McLean (1985, 1988) and McAfee and Reny (1992), our main results give conditions on beliefs under which full extraction or weak full extraction holds. We then consider the genericity of these conditions. We also consider natural restrictions under which the designer is limited to a single contract, and under which the designer typically will be limited in ability to extract information rents.

We close this section by recalling some standard definitions and notation for sets in ${\bf R}^n$. For a set $A\subseteq {\bf R}^n$, $\overline{\mbox{co}}(A)$ denotes the (closed) convex hull of $A$. For $A, B\subseteq {\bf R}^n$ with $A\subseteq B$, $\mbox{rint} A$ denotes the relative interior of $A$, relative to a superset $B$. If $B$ is a finite set with $k$ elements and $A \subseteq \Delta(B)$, $\mbox{rint} A$ denotes the relative interior of $A$, relative to $\{ x\in {\bf R}^k: \sum_i x_i = 1\}$; thus $\mbox{rint} A = A \cap {\bf R}^k_{++}$. 

We say a set $A\subseteq {\bf R}^n$ has \emph{full dimension} if 
\[
a\cdot x = 0 \ \ \forall a\in A \Rightarrow x=0
\]
Note that $A \subseteq {\bf R}^n$ has full dimension if $A$ contains a set $\{ a_1, \ldots ,a_n\}$ of $n$ linearly independent elements. 
In particular, if $A\subseteq \Delta(S)$ is convex, then $A$ has full dimension if and only if $\mbox{rint} A \not= \emptyset$. 

Finally, as a notational shorthand, throughout the paper we use a constant $r\in {\bf R}$ interchangeably with the corresponding  constant vector $(r, r, \ldots , r) \in {\bf R}^n$. 

\subsection{Knightian Uncertainty and Incomplete Preferences}

In this section, we briefly discuss the model of decision making in the presence of ambiguity or
Knightian uncertainty that gives a behavioral foundation to our notion of robustness. The connection between robustness and ambiguity comes from Knightian decision theory, developed in Bewley (1986), in which ambiguity is modeled by
incomplete preferences.\footnote{See also Aumann (1962, 1964) and, more recently, Dubra et al. (2004), Ok (2002), Shapley and Baucells (2008), Ghirardato et al. (2003), Gilboa et al. (2010), and Girotto and Holzer (2005).} We begin by briefly describing the framework and results.

Bewley (1986) axiomatizes incomplete preference relations that can be represented by a family of
subjective expected utility functions. The main result in Bewley (1986) shows that a
strict preference relation that is not necessarily complete, but satisfies other standard axioms of subjective expected utility, can be represented using a von Neumann-Morgenstern utility index and a family of probability distributions, with strict preference corresponding
to unanimous ranking according to this family. Here we adopt a slightly modified version of Bewley preferences that
uses weak preference as a primitive.

To illustrate in more detail, let $x, y\in {\bf R}^S_+$ denote state-contingent consumption bundles and let $\succsim$ denote a weak preference relation on ${\bf R}^S_+$. Bewley (1986) axiomatizes preference relations that can be represented by a closed, convex set $\Pi\subseteq \Delta(S)$ and a continuous, concave function $u:{\bf R}_+ \to {\bf R}$, unique up
to positive affine transformations, as follows:
\[
x\succsim y \mbox{ if and only if } \sum_{s\in S} \pi_s u(x_s) \geq \sum_{s\in S} \pi_s u(y_s) \ \ \forall \pi\in \Pi
\]
If $\succsim$ is complete, the set $\Pi$ reduces to a singleton and the standard subjective expected
utility representation obtains. Incompleteness is thus characterized by
multiplicity of beliefs. If $\succsim$ is not complete, comparisons between alternatives
are carried out one probability distribution at a time: one bundle is strictly preferred to
another if and only if its expected utility is larger under every probability distribution in
the set $\Pi$.

This representation captures the Knightian distinction between risk and ambiguity,
where an event is risky if its probability is known, and ambiguous otherwise. The decision
maker perceives only risk when $\Pi$ is a singleton, and perceives ambiguity otherwise. Thus incompleteness and ambiguity are two sides of the same phenomenon in this model: the amount of ambiguity the decision maker perceives and the degree of incompleteness of her
preference relation $\succsim$ are both measured by the size of the set $\Pi$.\footnote{For precise results along these lines, see Ghirardato et al. (2004) or Rigotti and Shannon (2005).}

An agent might choose an alternative because no other feasible option is preferred to it, or because it is
preferred to all other feasible options. For a general incomplete preference relation, these  notions need not coincide. We recall the standard definitions of maximal and optimal choices, reflecting this potential distinction, next. 

\begin{definition}
Given a preference relation $\succsim$ on ${\bf R}^S$ and subset $X\subseteq {\bf R}^S$, an element $x\in X$ is \emph{maximal} for $\succsim$ in $X$ if there exists no $y\in X$ such that $y\succ x$. An element $x\in X$ is \emph{optimal} for $\succsim$ in $X$ if $x\succsim y$ for all $y\in X$. 
\end{definition}

For Bewley preferences, these notions have particularly simple characterizations using the family of expected utility functions that represent the preference relation. In this case, an element $x\in X\subseteq {\bf R}^S$ is maximal for $\succsim$ in $X$ if and only if 
\[
\forall y\in X \ \exists \pi\in \Pi \mbox{ such that } \sum_{s\in S} \pi_s u(x_s) \geq \sum_{s\in S} \pi_s u(y_s) 
\]
An element $x\in X$ is optimal for $\succsim$ in $X$ if instead 
\[
\sum_{s\in S} \pi_s u(x_s) \geq \sum_{s\in S} \pi_s u(y_s) \ \ \forall y\in X, \ \ \forall \pi\in \Pi
\] 

In our problem, we take the preference relation $\succsim$ to have a particularly simple representation. For type $t\in T$, we assume that for any pair of stochastic contracts $x, y\in {\bf R}^S$, 
\[
x\succsim_t y \mbox{ if and only if } v(t) - \pi\cdot x \geq v(t) -\pi\cdot y \ \ \forall \pi\in \Pi(t)
\]
or equivalently, 
\[
x\succsim_t y \mbox{ if and only if }  \pi\cdot x \leq \pi\cdot y \ \ \forall \pi\in \Pi(t)
\]
For agents choosing from a menu of stochastic contracts $C  \subseteq {\bf R}^S$, a contract $c^m\in C$ is maximal for type $t\in T$ in $C$ if there is no contract $c\in C$ such that   
\[
\pi\cdot c < \pi\cdot c^m \ \ \forall \pi\in \Pi(t)
\]
In this case, it is possible that there are other feasible contracts in $C$ that have lower expected cost than $c^m$ for some beliefs and higher expected cost than $c^m$ for other beliefs. A contract $c^o\in C$ is optimal for type $t\in T$ in $C$ if instead 
\[
\pi\cdot c^o \leq \pi\cdot c \ \ \forall \pi\in \Pi(t)
\]
In this case, no other contract in $C$ has lower expected cost than $c^o$ for any belief in $\Pi(t)$. 

In any arbitrary finite or compact menu $C\subseteq {\bf R}^S$, every type $t\in T$ will have maximal choices. Even in a finite menu $C$, however, types need not have optimal choices. Offering a menu in which types have optimal choices thus can be a constraint on the designer. We show in section 4 that this constraint does not preclude full extraction under some natural conditions, but that this becomes a significant constraint under other natural conditions, possibly even limiting the designer to offering a single contract. Both results are illustrated in the following section, by means of a simple example.

\section{Leading Example}

We start with an example to illustrate the main ideas and results. 

For the example, we start with the standard problem of Cr{\' e}mer-McLean and McAfee-Reny in which the set $\Pi(t)$ is a singleton for each $t$. Let $\pi(t) \in \Delta(S)$ denote this single element for each $t\in T$, and suppose agents' beliefs $\{ \pi(t) : t\in T\}$ satisfy the standard  {\it convex independence} condition of Cr{\' e}mer-McLean (1988): for any $ t\in T$,
\[
\pi( t) = \sum_{s\in T} \mu_s \pi(s) \mbox{ for some } \mu\in \Delta(T) \Rightarrow \mu_t = 1
\]
Equivalently, convex independence requires that for each $ t\in T$, 
\[
\pi( t) \not\in \overline{\mbox{co}} \{ \pi(s) : s\not= t\}
\]
where, as noted in section 2 above, $\overline{\mbox{co}}(A)$ denotes the closed convex hull of the set $A\subseteq {\bf R}^n$. Cr{\' e}mer-McLean (1988) show that convex independence is necessary and sufficient for full extraction. In addition, whenever $\vert S\vert \geq \vert T \vert$, this condition is satisfied generically in many settings. This second result, which implies that full extraction holds generically, is  what gives the results of Cr{\' e}mer-McLean (1988) and McAfee-Reny (1992) so much power in many ways, rather than the characterization results alone. 

To set the stage and give some intuition for the robust versions of these constructions we introduce, we first recall the standard argument for full extraction given convex independence. Suppose beliefs $\{ \pi(t) : t\in T\}$ satisfy convex independence. Fix $t\in T$. Since $\pi( t) \not\in \overline{\mbox{co}} \{ \pi(s) : s\not= t\}$, the separating hyperplane theorem implies that there exists $z(t) \in {\bf R}^S$ such that 
\[
\pi(t) \cdot z(t) = 0
\]
and
\[
\pi(s) \cdot z(t) >0 \ \ \ \forall s\not= t
\]
See Figure 1. Now consider a contract of the form $c(t) = v(t) + \alpha(t) z(t)$ that requires the constant payment $v(t)$ and a stochastic payment equal to some scaled version of $z(t)$. For type $t$, this contract has expected cost $v(t)$:
\[
\pi(t) \cdot c(t) = v(t) + \alpha(t) ( \pi(t)\cdot z(t) ) = v(t)
\]
while for types $s\not= t$, the expected cost is
\[
\pi(s) \cdot c(t) = v(t) + \alpha(t) ( \pi(s) \cdot z(t) ) 
\]
To make this contract unattractive to types $s\not= t$,  set $\alpha(t) >0$ sufficiently large so that
\[
\alpha (t) > \sup_{s\not= t} \frac{v(s) - v(t)}{\pi(s) \cdot z(t)}
\]
Since $\pi(s)\cdot z(t)>0$ for all $s\not= t$ and $T$ is finite, such an $\alpha(t)>0$ exists. Then for $s\not= t$, 
\[
\pi(s) \cdot c(t) = v(t) + \alpha(t) ( \pi(s) \cdot z(t) ) > v(s)
\]
by construction. Repeating this construction for each $t\in T$ yields the collection $C = \{ c(t): t\in T\}$ which achieves full extraction. 


\begin{figure}
\centering
\begin{minipage}{0.45\textwidth}
\centering
\begin{tikzpicture}[scale=0.70]
\tikzstyle{every node}=[font=\small]
\draw (1,0) -- (6,7) -- (11,0) -- cycle;

\fill (4,2) circle (1pt) node[left] {$\pi(t_1) \ $};

\fill (6,5) circle (1pt) node[above] {$\pi(t_2)$};

\fill (8,3) circle (1pt) node [below right ] {$\pi(t_3)$};

\fill (6,1) circle (1pt) node [below] {$\pi(t_4)$};

\end{tikzpicture}
\end{minipage}\hfill
\begin{minipage}{0.45\textwidth}
\begin{tikzpicture}[scale=0.70]
\tikzstyle{every node}=[font=\small]
\draw (1,0) -- (6,7) -- (11,0) -- cycle;

\fill (4,2) circle (1pt) node[left] {$\pi(t_1) \ $};

\fill (6,5) circle (1pt) node[above] {$\pi(t_2)$};

\fill (8,3) circle (1pt) node [below right ] {$\pi(t_3)$};

\fill (6,1) circle (1pt) node [below] {$\pi(t_4)$};

\coordinate (t1) at (4,2);
\coordinate (t2) at (6,5);
\coordinate (t3) at (8,3);
\coordinate (t4) at (6,1);

\path[draw, fill=gray!30]  (t2) -- (t3) -- (t4) -- cycle;
\end{tikzpicture}
\end{minipage}
\caption{Beliefs $\{ \pi(t)\in \Delta(S) : t\in T\}$ satisfying convex independence.}
\end{figure}
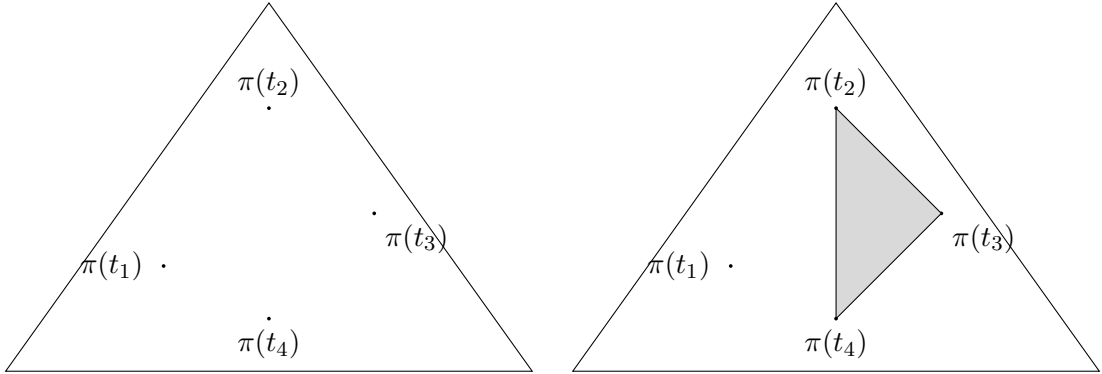


Now let $\varepsilon >0$, and suppose the uncertainty in the model is captured by $\varepsilon$-contamination of the original beliefs $\{ \pi(t): t\in T\}$. 
That is, for each $t\in T$, let 
\[
\Pi_\varepsilon(t) = \{ \pi\in \Delta(S): \pi = (1-\varepsilon)\pi(t) + \varepsilon \pi', \ \ \pi'\in \Delta(S)\}
\]
In this case, to extract all of the surplus from a given type $t\in T$ the designer must choose a set of contracts $C = \{ c(t) \in {\bf R}^S: t\in T\}$ with the property that for each $t\in T$, 
\[
\pi \cdot c(t) \leq \pi \cdot c(s) \ \ \forall \pi \in \Pi_\varepsilon(t)
\]
and
\[
v(t) - \pi\cdot c(t) \geq 0 \ \ \forall \pi \in \Pi_\varepsilon(t) \mbox{ with } \pi\cdot c(t) = v(t) \ \mbox{ for some } \pi\in \Pi_\varepsilon(t)
\]
Given that the beliefs $\{ \pi(t) : t\in T\}$ satisfy convex independence, for $\varepsilon$ sufficiently small full extraction is still possible. We can establish this by a natural modification of the standard argument, based on the observation that for $\varepsilon >0$ sufficiently small, for each $t\in T$, 
\[
\Pi_\varepsilon(t) \cap (\overline{\mbox{co}}(\cup_{s\not= t} \Pi_\varepsilon(s)) = \emptyset
\]
See Figure 2. Then, mirroring the previous argument, for each $t\in T$ there exists $z(t) \in {\bf R}^S$ such that 
\[
\pi\cdot z(t) \leq 0 \ \ \forall \pi\in \Pi_\varepsilon(t)
\]
and
\[
\pi\cdot z(t) > 0 \ \ \forall \pi\in \Pi_\varepsilon(s), \ \ \forall s\not= t
\]
Note that $z(t)$ can be chosen so that in addition, $\pi\cdot z(t) = 0$ for some $\pi\in \Pi_\varepsilon(t)$; we assume $z(t)$ has been so chosen. 


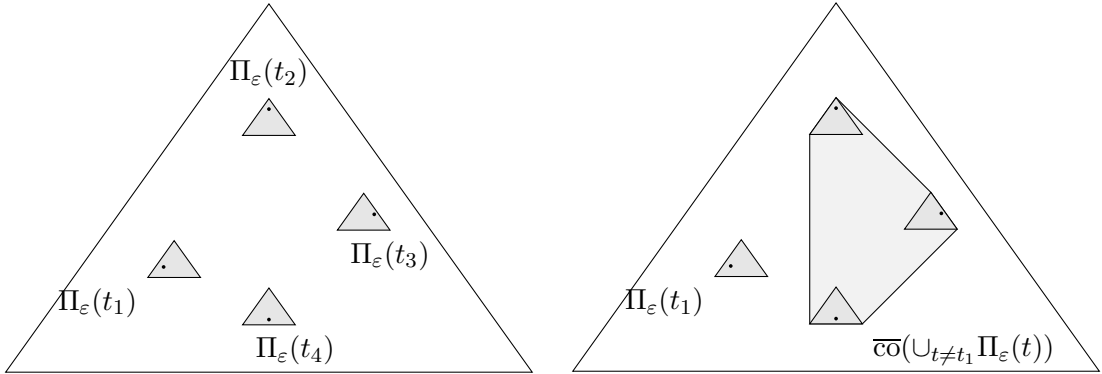
\begin{figure}
\centering
\begin{minipage}{0.45\textwidth}
\centering
\begin{tikzpicture}[scale=0.70]
\tikzstyle{every node}=[font=\small]

\draw (1,0) -- (6,7) -- (11,0) -- cycle;

\path[draw, fill=gray!20] (3.70, 1.8) node[below left] {$\Pi_\varepsilon(t_1)$} -- (4.2, 2.5) -- (4.7, 1.8) -- cycle;
\path[draw, fill=gray!20] (5.5, 4.5) -- (6, 5.2) node [above] {$\Pi_\varepsilon(t_2)$} -- (6.5, 4.5) -- cycle;
\path[draw, fill=gray!20] (7.3, 2.7) -- (7.8, 3.4) -- (8.3, 2.7) node[below] {$\Pi_\varepsilon(t_3)$} -- cycle;
\path[draw, fill=gray!20] (5.5, .9) -- (6, 1.6) -- (6.5, .9) node[below] {$\Pi_\varepsilon(t_4)$} -- cycle;

\fill (4,2) circle (1pt);

\fill (6,5) circle (1pt);

\fill (8,3) circle (1pt);

\fill (6,1) circle (1pt);

\end{tikzpicture}
\end{minipage}\hfill
\begin{minipage}{0.45\textwidth}
\begin{tikzpicture}[scale=0.70]
\tikzstyle{every node}=[font=\small]

\draw (1,0) -- (6,7) -- (11,0) -- cycle;

\path[draw, fill=gray!10] (5.5, .9) -- (5.5, 4.5) -- (6, 5.2) -- (7.8, 3.4) -- (8.3, 2.7) -- (6.5, .9) -- cycle; 

\path[draw, fill=gray!20] (3.70, 1.8) node[below left] {$\Pi_\varepsilon(t_1)$} -- (4.2, 2.5) -- (4.7, 1.8) -- cycle;
\path[draw, fill=gray!20] (5.5, 4.5) -- (6, 5.2)  -- (6.5, 4.5) -- cycle;
\path[draw, fill=gray!20] (7.3, 2.7) -- (7.8, 3.4) -- (8.3, 2.7)  -- cycle;
\path[draw, fill=gray!20] (5.5, .9) -- (6, 1.6) -- (6.5, .9) node[below right] {$\overline{\mbox{co}}(\cup_{t\not= t_1}\Pi_\varepsilon(t)$)} -- cycle;

\fill (4,2) circle (1pt);

\fill (6,5) circle (1pt);

\fill (8,3) circle (1pt);

\fill (6,1) circle (1pt);

\end{tikzpicture}
\end{minipage}
\caption{Beliefs $\{ \Pi_\varepsilon(t)\subseteq \Delta(S) : t\in T\}$ satisfying convex independence.}
\end{figure}


As above, consider a contract of the form $c(t) = v(t) + \alpha(t) z(t)$ that requires the constant payment $v(t)$ and a stochastic payment equal to  some scaled version of $z(t)$. For type $t$, for any $\alpha(t) \geq 0$ this contract has expected cost no more than $v(t)$:
\[
\pi(t) \cdot c(t) = v(t) + \alpha(t) ( \pi(t)\cdot z(t) ) \leq v(t) \ \ \forall \pi\in \Pi_\varepsilon(t)
\]
while for types $s\not= t$, the expected cost is
\[
\pi(s) \cdot c(t) = v(t) + \alpha(t) ( \pi(s) \cdot z(t) ) \ \ \forall \pi\in \Pi_\varepsilon(s)
\]
To make this contract unattractive to types $s\not= t$,  set $\alpha(t) >0$ sufficiently large so that
\[
\alpha (t) > \sup_{\substack{{\pi \in \Pi_\varepsilon(s)}\\ {s\not= t}}} \frac{v(s) - v(t)}{\pi \cdot z(t)}
\]
Since $\pi\cdot z(t)>0$ for all $\pi\in \Pi_\varepsilon(s)$ and $s\not= t$, and $T$ is finite, such an $\alpha(t)>0$ exists. Then for $s\not= t$, 
\[
\pi(s) \cdot c(t) = v(t) + \alpha(t) ( \pi(s) \cdot z(t) ) > v(s) \ \ \forall \pi\in \Pi_\varepsilon(s)
\]
by construction. For type $t$, 
\[
\pi(t)\cdot c(t) \leq v(t) \ \ \forall \pi\in \Pi_\varepsilon(t), \ \mbox{ with } \pi\cdot c(t) = v(t) \ \mbox{ for some } \pi\in \Pi_\varepsilon(t)
\]
Repeating this construction for each $t\in T$ yields the collection $C = \{ c(t): t\in T\}$ which achieves full extraction. 

A second immediate observation is that for $\varepsilon >0$ sufficiently large, full extraction becomes impossible. This follows because for sufficiently large $\varepsilon$, the sets $\{ \Pi_\varepsilon(t) : t\in T\}$ will overlap sufficiently to make full extraction impossible. In particular, for $\varepsilon $ sufficiently large, for some $t\in T$, 
\[
\overline{\mbox{co}}(\cup_{s\not= t} \Pi_\varepsilon(s)) \subseteq \Pi_\varepsilon(t) 
\]
In this case, there is some $s\not= t$ such that $\Pi_\varepsilon(s) \subseteq \Pi_\varepsilon(t)$. To see that full extraction then is impossible, suppose $v(s) > v(t)$. Then for any contract $c(t)$ for which 
\[
\pi\cdot c(t) \leq v(t) \ \ \forall \pi\in \Pi_\varepsilon(t)
\]
it must also be the case that 
\[
\pi\cdot c(t) \leq v(t) < v(s) \ \ \forall \pi\in \Pi_\varepsilon(s)
\]
Thus it is impossible for the designer to achieve full extraction  in this case. 

Starting from a fixed set of beliefs $\{ \pi(t): t\in T\}$ satisfying convex independence, full extraction then remains possible for a degree of uncertainty $\varepsilon >0$ sufficiently small, but eventually becomes impossible for $\varepsilon$ sufficiently large. Given a fixed degree of uncertainty $\varepsilon >0$, a designer might not be able to achieve full extraction, even when the original beliefs $\{ \pi(t): t\in T\}$ satisfy convex independence. In that case, incentive compatibility can impose strong constraints on the designer. In particular, given $\varepsilon >0$ there is an open set of beliefs $\{\pi(t): t\in T\} \subseteq \Delta(S)^T$ such that for the corresponding uncertain beliefs $\{ \Pi_\varepsilon(t): t\in T\}$, incentive compatibility means the designer can offer only a single contract. 

To see this, let $\varepsilon >0$ be fixed. Suppose $C = \{ c(t) \in {\bf R}^S: t\in T\}$ is \emph{incentive compatible}, that is, for each $t\in T$, 
\[
v(t) - \pi(t) \cdot c(t) \geq v(t) - \pi(t) \cdot c(s) \ \ \forall \pi\in \Pi_\varepsilon(t), \ \ \forall s\not= t
\]
Then for a pair of types $t_1, t_2 \in T$, it must be that 
\[
\pi\cdot c(t_1) \leq \pi\cdot c(t_2) \ \ \forall \pi\in \Pi_\varepsilon(t_1)
\]
and
\[
\pi\cdot c(t_2) \leq \pi\cdot c(t_1) \ \ \forall \pi\in \Pi_\varepsilon(t_2)
\]
Thus 
\[
\pi\cdot (c(t_1) - c(t_2))\leq 0 \ \ \forall \pi\in \Pi_\varepsilon(t_1)
\]
and 
\[
\pi\cdot (c(t_2) - c(t_1)) \leq 0 \ \ \forall \pi\in \Pi_\varepsilon(t_2)
\]
Putting these together yields
\[
\pi\cdot (c(t_1) - c(t_2)) = 0 \ \ \forall \pi\in \Pi_\varepsilon(t_1)\cap \Pi_\varepsilon(t_2)
\]
If $\Pi_\varepsilon(t_1)\cap \Pi_\varepsilon(t_2)$ has full dimension, this implies $c(t_1) - c(t_2) = 0$, that is, $c(t_1) = c(t_2)$.  If this is true for any pair $t, t'\in T$, then $C$ must consist of a single contract. We show below that given $\varepsilon>0$, there is an open set in $\Delta(S)^T$ satisfying convex independence such that $\Pi_\varepsilon(t)\cap \Pi_\varepsilon(t')$ has full dimension for any pair $t, t'\in T$. See Figure 3. 


\begin{figure}
\centering
\begin{tikzpicture}

\draw (1,0) -- (6,7) -- (11,0) -- cycle;

\path[draw, fill=gray!10] (2.8, 1.2)  -- (4.8, 4) -- (6.8, 1.2) -- cycle;
\path[draw, fill=gray!10] (4, 3) -- (6, 5.8)  -- (8, 3) -- cycle;
\path[draw, fill=gray!10] (5.2, 1.8) -- (7.2,4.6) -- (9.2,1.8)  -- cycle;
\path[draw, fill=gray!10] (4, .6) -- (6, 3.4) -- (8, .6)  -- cycle;

\path[draw] (2.8, 1.2) node[below] {$\Pi_\varepsilon(t_1)$} -- (4.8, 4) -- (6.8, 1.2) -- cycle;
\path[draw] (4, 3) -- (6, 5.8) node [above] {$\Pi_\varepsilon(t_2)$} -- (8, 3) -- cycle;
\path[draw] (5.2, 1.8) -- (7.2,4.6) -- (9.2,1.8) node[below] {$\Pi_\varepsilon(t_3)$} -- cycle;
\path[draw] (4, .6) -- (6, 3.4) -- (8, .6) node[below] {$\Pi_\varepsilon(t_4)$} -- cycle;

\fill (4,2) circle (1pt);

\fill (6,5) circle (1pt);

\fill (8,3) circle (1pt);

\fill (6,1) circle (1pt);

\end{tikzpicture}

\caption{$\Pi_\varepsilon(t)\cap \Pi_\varepsilon(t')$ has full dimension for all $t,t'$. }
\end{figure}
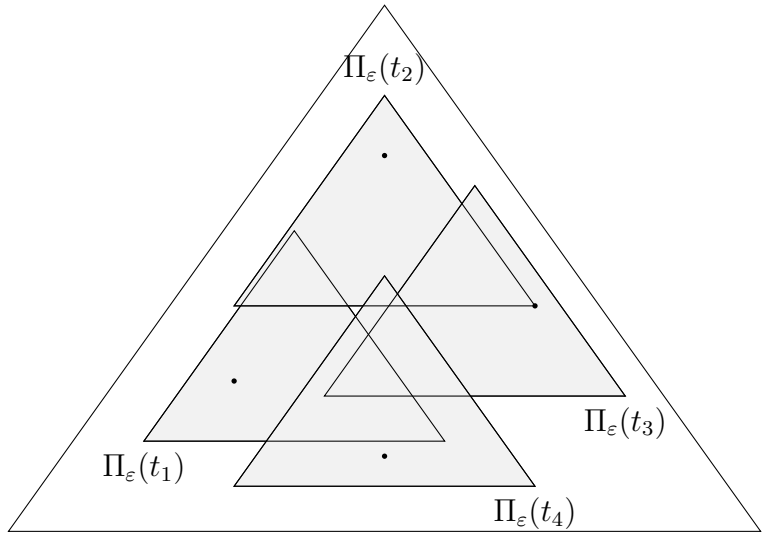


Thus for any such beliefs, any incentive compatible collection $C$ contains just one contract $c$. This observation also provides an upper bound on the designer's expected revenue, given additional information about the designer's beliefs. In particular, let $\pi(d)\in \Delta(S)$ denote a belief of the designer. 
Without loss of generality, suppose $v(t_1) $ is the smallest value among the types in $T$, and write the contract $c$ as $c = v(t_1) + z$ for $z\in {\bf R}^S$. Now note that individual rationality requires
\[
\pi\cdot c = v(t_1) + \pi\cdot z \leq v(t_1) \ \ \forall \pi\in \Pi_\varepsilon(t_1)
\]
Thus $\pi\cdot z \leq 0 $ for all $\pi\in \Pi_\varepsilon(t_1)$. Say there is {\it concurrence} if the designer's belief $\pi(d) \in  \Pi_\varepsilon(t_1)$. In this case, we can give a simple upper bound on the designer's expected revenue from any menu satisfying incentive compatibilty and individual rationality: 
\[
T\pi(d) \cdot c \leq T v(t_1)
\]
In this case, an optimal menu for the designer is to offer just the deterministic contract $v(t_1)$. Note that this conclusion also follows if the designer uses a maxmin criterion, as long as the designer uses a set of beliefs $\Pi(d) \subseteq \Delta(S)$ such that $\Pi(d) \cap \Pi_\varepsilon(t_1) \not= \emptyset$. 

Finally, for any fixed upper bound on the degree of uncertainty $\varepsilon >0$, there is an open set of beliefs $\{\pi(t): t\in T\} \subseteq \Delta(S)^T$ satisfying convex independence such that full extraction is not possible for the corresponding uncertain beliefs $\{ \Pi_{\varepsilon_t}(t): t\in T\}$, where $0<\varepsilon_t \leq \varepsilon$ for each $t$. 

We collect these observations in the two propositions below, and provide the additional proofs not given above. 

\begin{proposition}
\begin{enumerate}[(i)]
\item For every $\{ \pi(t): t\in T\} \subseteq \Delta(S)^T$ satisfying convex independence, there exists $\varepsilon >0$ sufficiently small such that full extraction is possible for $\{ \Pi_\varepsilon(t): t\in T\}$.
\item For every   $\{ \pi(t): t\in T\} \subseteq \Delta(S)^T$  there exists $\varepsilon >0$ sufficiently large such that full extraction is impossible for $\{ \Pi_\varepsilon(t): t\in T\}$.
\item For every $\varepsilon >0$, there is an open subset $O_\varepsilon \subseteq \Delta(S)^T$ satisfying convex independence such that for all $\{ \pi(t): t\in T\} \in O_\varepsilon$, full extraction is not possible for some corresponding uncertain beliefs $\{\Pi_{\varepsilon_t}(t): t\in T\}$ with $0<\varepsilon_t\leq \varepsilon$ for each $t\in T$. 
\end{enumerate}
\end{proposition}
\begin{proof}
For (i), fix $\{ \pi(t): t\in T\} \subseteq \Delta(S)^T$ satisfying convex independence. Fix $t\in T$. By convex independence, 
\[
\pi(t) \not\in \overline{\mbox{co}}\{ \pi(s) : s\not= t\}
\]
Then there exists $\varepsilon >0$ sufficiently small such that for $\pi'(t) \in \Pi_\varepsilon(t)$ and $\pi'(s)\in \Pi_\varepsilon(s), s\not= t$, 
\[
\pi'(t) \not\in \overline{\mbox{co}}\{ \pi'(s) : s\not= t\}
\]
that is,
\[
\Pi_\varepsilon(t) \cap \overline{\mbox{co}} (\cup_{s\not= t} \Pi_\varepsilon(s)) = \emptyset
\]
Repeating this argument for each $t$ and using the finiteness of $T$ establishes the claim.

For (iii), fix $\varepsilon >0$ and choose $\{ \pi(t)\in \Delta(S): t\in T\}$ satisfying convex independence with $\pi(t)\gg 0$ for each $t\in T$, and such that $\pi(t_1) \in \mbox{rint} \Pi_\varepsilon(t_2)$. Then there exists $\varepsilon_{t_1} >0$ such that $\Pi_{\varepsilon_{t_1}}(t_1) \subseteq \Pi_\varepsilon(t_2)$. Setting $\varepsilon_t = \varepsilon$ for $t\not= t_1$, then by the argument above, full extraction is not possible for the beliefs $\{ \Pi_{\varepsilon_t}(t): t\in T\}$. Moreover, there is an open set in $\Delta(S)^T$ containing $\{ \pi(t)\in \Delta(S): t\in T\}$ satisfying convex independence on which the same argument holds. 
\end{proof}

\begin{proposition}
For each $\varepsilon>0$, there is an open subset $O_\varepsilon \subseteq \Delta(S)^T$ satisfying convex independence such that for all $\{ \pi(t): t\in T\} \in O_\varepsilon$, a menu $C= \{ c(t) \in {\bf R}^S: t\in T\}$ is incentive compatible for $\{ \Pi_\varepsilon(t): t\in T\}$ only if $c(t) = c(t')$ for all $t, t'\in T$. 

Given $\varepsilon>0$ and $\{ \pi(t): t\in T\} \in O_\varepsilon$, if the designer's beliefs are concurrent with $\Pi_\varepsilon(t_1)$ where $v(t_1) = \min_{t\in T} v(t)$, then the designer's expected revenue from any incentive compatible and individually rational menu $C$ is less than or equal to $Tv(t_1)$, and the deterministic contract $v(t_1)$ is optimal for the designer. 
\end{proposition}
\begin{proof}
The second claim follows from the argument given in the text above. For the first claim, fix $\varepsilon >0$. Let $t_1\in T$ be fixed. Choose $\pi(t_1)\in \Delta(S)$ such that $\Pi_{\varepsilon}(t_1) \subseteq \mbox{rint} \Delta(S)$.  Then choose $\pi(s), s\not=t_1$ such that  $\{ \pi(t)\in \Delta(S): t\in T\}$ satisfies convex independence and such that $\pi(t_1)\in \mbox{rint}\Pi_\varepsilon(s)$ for each $s\not= t_1$. Then for each $s$ there exists $\delta_s>0$ such that $B_{\delta_s}(\pi(t_1)) \subseteq \Pi_\varepsilon(s)$, where for $\pi\in \Delta(S)$ and $\beta >0$, $B_\beta(\pi)$ denotes the ball of radius $\beta$ around $\pi$ in $\Delta(S)$, so $B_\beta(\pi) = \{ \pi'\in \Delta(S): \Vert \pi' -\pi\Vert <\beta\}$. Setting $\delta = \min_{s}\delta_s >0$, this implies $B_\delta(\pi(t_1)) \subseteq \cap_{t\in T} \Pi_\varepsilon(t)$. In particular, for any $t, t'\in T$,  $B_\delta(\pi(t_1)) \subseteq \Pi_\varepsilon(t)\cap \Pi_\varepsilon(t')$. Since $B_\delta(\pi(t_1)) $ has full dimension, this shows that  $\Pi_\varepsilon(t)\cap \Pi_\varepsilon(t')$ has full dimension for any $t, t'\in T$. Moreover, there exists $\alpha>0$ sufficiently small such that  the same argument applies to any $\{ \pi'(t)\in \Delta(S): t\in T\}$ with $\pi'(t) \in B_\alpha(\pi(t))$ for each $t\in T$. 
\end{proof}

\section{Surplus Extraction}

In this section we turn to the general model. We first give analogues of the classic results of Cr{\' e}mer-McLean (1988) and McAfee and Reny (1992) for the general setting in which beliefs can be arbitrary closed, convex sets, based on the notions of full extraction and weak full extraction. Next we connect these notions to choice behavior under Knightian uncertainty, and provide alternative results for notions of full extraction motivated by such behavior. We then consider limits on the designer's ability to extract information rents when these results do not apply. In particular, we show that when types' beliefs are sufficiently overlapping, any incentive compatible  menu of contracts contains a unique contract. 

We start by developing two versions of the classic convex independence condition for sets of beliefs.  

\begin{definition}
Beliefs $\{ \Pi(t)\subseteq \Delta(S) : t\in T\}$ satisfy \emph{convex independence} if 
\[
\Pi(t) \cap \overline{\mbox{co}}( \cup_{s\not=t} \Pi(s)) = \emptyset \ \ \forall t\in T
\]
\end{definition}

\bigskip

\begin{definition}
Beliefs $\{ \Pi(t)\subseteq \Delta(S) : t\in T\}$ satisfy \emph{weak convex independence} if there exists $\{ \pi(t)\in \Pi(t): t\in T\}$ satisfying convex independence, that is, such that
\[
\pi(t) \not\in \overline{\mbox{co}} \{ \pi(s): s\not= t\} \ \ \forall t\in T
\]
\end{definition}

\bigskip

Note that convex independence for the collection $\{ \Pi(t): t\in T\}$ is equivalent to the requirement that every selection $\{\pi(t)\in \Delta(S): \pi(t) \in \Pi(t) \ \forall t\in T\}$ satisfies convex independence, while weak convex independence requires just that some  such selection satisfies convex independence. When the set $\Pi(t)$ is a singleton for each $t\in T$, these notions are equivalent, and are equivalent to the standard condition of Cr{\' e}mer-McLean (1988) and McAfee-Reny (1992). 

Our first main results show that weak full extraction holds whenever beliefs satisfy weak convex independence, while full extraction holds whenever  beliefs satisfy convex independence. 

\begin{theorem}
If beliefs $\{ \Pi(t) \subseteq \Delta(S) : t\in T\}$ satisfy weak convex independence, then weak full extraction holds. 
\end{theorem}
\begin{proof}
Fix $v:T\to {\bf R}$. By weak convex independence, there exists a selection $\{ \pi(t) \in \Delta(S): \pi(t) \in \Pi(t) \ \ \forall t\in T\}$ satisfying convex independence. Then fix $t\in T$. Since $\pi( t) \not\in \overline{\mbox{co}} \{ \pi(s) : s\not= t\}$, there exists $z(t) \in {\bf R}^S$ such that 
\[
\pi(t) \cdot z(t) = 0
\]
and
\[
\pi(s) \cdot z(t) >0 \ \ \ \forall s\not= t
\]
Now consider a contract of the form $c(t) = v(t) + \alpha(t) z(t)$ that requires the constant payment $v(t)$ and a stochastic payment some scaled version of $z(t)$. For type $t$, this contract has expected cost $v(t)$:
\[
\pi(t) \cdot c(t) = v(t) + \alpha(t) ( \pi(t)\cdot z(t) ) = v(t)
\]
while for types $s\not= t$, the expected cost is
\[
\pi(s) \cdot c(t) = v(t) + \alpha(t) ( \pi(s) \cdot z(t) ) 
\]
To make this contract unattractive to types $s\not= t$,  set $\alpha(t) >0$ sufficiently large so that
\[
\alpha (t) > \sup_{s\not= t} \frac{v(s) - v(t)}{\pi(s) \cdot z(t)}
\]
Since $\pi(s)\cdot z(t)>0$ for all $s\not= t$ and $T$ is finite, such an $\alpha(t)>0$ exists. Then for $s\not= t$, 
\[
\pi(s) \cdot c(t) = v(t) + \alpha(t) ( \pi(s) \cdot z(t) ) > v(s)
\]
by construction. Repeating this construction for each $t\in T$ yields the collection $C = \{ c(t): t\in T\}$, which achieves weak full extraction. 
\end{proof}

\begin{theorem}
If beliefs $\{ \Pi(t) \subseteq \Delta(S) : t\in T\}$ satisfy convex independence, then full extraction holds. 
\end{theorem}
\begin{proof}
Fix $v:T\to {\bf R}$, and fix $t\in T$. By convex independence, 
\[
\Pi(t) \cap (\overline{\mbox{co}}(\cup_{s\not= t} \Pi(s)) = \emptyset
\]
Then there exists $z(t) \in {\bf R}^S$ such that 
\[
\pi\cdot z(t) \leq 0 \ \ \forall \pi\in \Pi(t)
\]
and
\[
\pi\cdot z(t) > 0 \ \ \forall \pi\in \Pi(s), \ \ \forall s\not= t
\]
Note that $z(t)$ can be chosen so that in addition, $\pi\cdot z(t) = 0$ for some $\pi\in \Pi(t)$; we assume $z(t)$ has been so chosen. 

Now consider a contract of the form $c(t) = v(t) + \alpha(t) z(t)$. For type $t$, for any $\alpha(t) \geq 0$ this contract has expected cost no more than $v(t)$:
\[
\pi(t) \cdot c(t) = v(t) + \alpha(t) ( \pi(t)\cdot z(t) ) \leq v(t) \ \ \forall \pi\in \Pi(t)
\]
while for types $s\not= t$, the expected cost is
\[
\pi(s) \cdot c(t) = v(t) + \alpha(t) ( \pi(s) \cdot z(t) ) \ \ \forall \pi\in \Pi(s)
\]
To make this contract unattractive to types $s\not= t$,  set $\alpha(t) >0$ sufficiently large so that
\[
\alpha (t) > \sup_{\substack{{\pi \in \Pi(s)}\\ {s\not= t}}} \frac{v(s) - v(t)}{\pi \cdot z(t)}
\]
Since $\pi\cdot z(t)>0$ for all $\pi\in \Pi(s)$ and $s\not= t$, and $T$ is finite, such an $\alpha(t)>0$ exists. Then for $s\not= t$, 
\[
\pi(s) \cdot c(t) = v(t) + \alpha(t) ( \pi(s) \cdot z(t) ) > v(s) \ \ \forall \pi\in \Pi(s)
\]
by construction. For type $t$, 
\[
\pi(t)\cdot c(t) \leq v(t) \ \ \forall \pi\in \Pi(t), \ \mbox{ with } \pi\cdot c(t) = v(t) \ \mbox{ for some } \pi\in \Pi(t)
\]
Repeating this construction for each $t\in T$ yields the collection $C = \{ c(t): t\in T\}$, which achieves full extraction. 
\end{proof}

Next, we identify a condition on beliefs that is necessary for full extraction. Say beliefs $\{ \Pi(t) \subseteq \Delta(S): t\in T\}$ satisfy {\it convex dependence} if for some $t\in T$:
\[
\overline{\mbox{co}} (\cup_{s\not= t} \Pi(s)) \subseteq \Pi(t)
\]
As we show next, a necessary condition for full extraction is that beliefs do not satisfy convex dependence, that is, that for all $t\in T$, 
\[
\overline{\mbox{co}} (\cup_{s\not= t} \Pi(s)) \not\subseteq \Pi(t)
\]

\begin{theorem}
Full extraction holds only if beliefs $\{ \Pi(t)\subseteq \Delta(S): t\in T\}$ do not satisfy convex dependence. 
\end{theorem}
\begin{proof}
Suppose beliefs satisfy convex dependence, so for some $t_0\in T$,
\[
\overline{\mbox{co}} (\cup_{s\not= t_0} \Pi(s)) \subseteq \Pi(t_0)
\]
Let $v(t)> v(t_0)$ for all $t\not= t_0$. Suppose by way of contradiction that the menu  $\{ c(t)\in {\bf R}^S: t\in T\}$ achieves full extraction. Then 
for $t_0$, 
\[
v(t_0) - \pi\cdot c(t_0) \geq 0 \ \forall \pi\in \Pi(t_0)
\]
while for $t\not= t_0$, 
\[
v(t) - \pi\cdot c(t_0) \leq 0 \ \forall \pi\in \Pi(t)
\]
But for $t\not= t_0$, 
\[
\Pi(t) \subseteq \overline{\mbox{co}} (\cup_{s\not= t_0} \Pi(s)) \subseteq \Pi(t_0)
\]
Thus for $t\not= t_0$ 
\begin{eqnarray*}
v(t) - \pi\cdot c(t_0) &=& v(t) - v(t_0) + v(t_0) -\pi\cdot c(t_0)\\
&\geq & v(t) - v(t_0) > 0 \ \ \forall \pi\in \Pi(t)
\end{eqnarray*}
This is a contradiction. Thus full extraction is not possible. 
\end{proof}

Next we connect the notions of full extraction and weak full extraction with choice behavior of agents with incomplete preferences. We start from the idea that agents are choosing from a finite menu $C \subseteq {\bf R}^S$, as above, and that each type $t\in T$ chooses according to the Bewley preference relation $\succsim_t$ as in section 2.2. We strengthen the notion of a contract as a maximal or optimal choice by allowing the agent to randomize over the elements in $C$. As we will see in the proofs of Corollaries 1 and 2 below, or as might be apparent by a careful consideration of the definitions of full extraction and weak full extraction above, this strengthening essentially comes for free. The conditions under which we showed extraction is possible when randomization is ruled out are the same as those guaranteeing extraction even allowing for randomization.

\begin{definition}
Let $C\subseteq {\bf R}^S$ be a finite menu of contracts. A contract $c^m\in C$ is \emph{mixed-strategy maximal} for type $t\in T$ in the menu $C$ if there is no mixed strategy $\sigma \in \Delta(C)$ such that 
\[
v(t) - \sum_{c\in C}\sigma(c)(\pi\cdot c) > v(t) - \pi\cdot c^m \ \ \forall \pi\in \Pi(t)
\]
\end{definition}

\begin{definition}
Let $C\subseteq {\bf R}^S$ be a finite menu of contracts. A contract $c^o\in C$ is \emph{mixed-strategy optimal} for type $t\in T$ in the menu $C$ if for all mixed strategies $\sigma \in \Delta(C)$ 
\[
v(t) -  \pi\cdot c^o \geq v(t) -  \sum_{c\in C}\sigma(c)(\pi\cdot c)\ \ \forall \pi\in \Pi(t)
\]
\end{definition}

Note that $c^m$ is mixed-strategy maximal for type $t$ in $C$ if and only if $c^m$ is maximal for $t$ in $\Delta(C)$; similarly $c^o$ is mixed-strategy optimal for $t$ if and only if $c^o$ is optimal for $t$ in $\Delta(C)$. 

Two natural notions of extraction then follow.

\begin{definition}
\emph{Optimal full extraction} holds if, given $v:T\to {\bf R}$, there exists a menu of contracts $C=\{ c(t) \in {\bf R}^S: t\in T\}$ such that for each $t\in T$, $c(t)$ is mixed-strategy optimal for $t$ in $C$ and 
\[
v(t) - \pi(t) \cdot c(t) \geq 0 \ \ \forall \pi\in \Pi(t), \mbox{ with } v(t) = \pi(t) \cdot c(t) \mbox{ for some } \pi\in \Pi(t)
\]
\end{definition}

\begin{definition}
\emph{Maximal full extraction} holds if, given $v:T\to {\bf R}$, there exists a menu of contracts $C=\{ c(t) \in {\bf R}^S: t\in T\}$ such that for each $t\in T$, $c(t)$ is mixed-strategy maximal for $t$ in $C$ and 
\[
v(t) - \pi(t) \cdot c(t) = 0  \mbox{ for some } \pi\in \Pi(t)
\]
\end{definition}

Note that full extraction implies optimal full extraction, and weak full extraction implies maximal full extraction. Thus we obtain the following  corollaries of Theorems 1, 2 and 3.

\begin{corollary}
If beliefs $\{ \Pi(t) \subseteq \Delta(S) : t\in T\}$ satisfy weak convex independence, then maximal full extraction holds. 
\end{corollary}
\begin{proof}
Let $v:T\to {\bf R}$ be given. By Theorem 1, there exists a menu $C= \{ c(t)\in {\bf R}^S: t\in T\}$ such that for each $t\in T$,
\[
v(t) - \pi(t) \cdot c(t) = 0 \mbox{ and } v(t) - \pi(t) \cdot c(s) \leq 0  \ \ \forall s\not= t
\]
for some $\pi\in \Pi(t)$. 
Thus 
\[
v(t) - \pi(t) \cdot c(s) \leq 0 \leq v(t) -\pi(t) \cdot c(t) \ \ \forall s\not= t
\]
Now for any mixed strategy $\sigma\in \Delta(C)$, 
\[
v(t) - \sum_{c\in C} \sigma(c) (\pi(t) \cdot c) = \sum_{c\in C} \sigma(c) ( v(t) - \pi(t)\cdot c)  \leq v(t) - \pi(t) \cdot c(t)
\]
Thus $c(t)$ is mixed-strategy maximal for $t$ in $C$. 
\end{proof}

\begin{corollary}
If beliefs $\{ \Pi(t) \subseteq \Delta(S) : t\in T\}$ satisfy convex independence, then optimal full extraction holds. 
\end{corollary}
\begin{proof}
Let $v:T\to {\bf R}$ be given. By Theorem 2, there exists a menu $C= \{ c(t)\in {\bf R}^S: t\in T\}$ such that for each $t\in T$,
\[
v(t) - \pi(t) \cdot c(t) \geq 0 \ \ \forall \pi\in \Pi(t), \mbox{ with } v(t) = \pi(t) \cdot c(t) \mbox{ for some } \pi\in \Pi(t)
\]
and 
\[
v(t) - \pi(t) \cdot c(s) \leq 0 \ \ \forall \pi\in \Pi(t), \ \ \forall s\not= t
\]
Thus 
\[
v(t) - \pi(t) \cdot c(s) \leq 0 \leq v(t) -\pi(t) \cdot c(t) \ \ \forall \pi\in \Pi(t) \ \ \forall s\not= t
\]
Now for any mixed strategy $\sigma\in \Delta(C)$, 
\[
v(t) - \sum_{c\in C} \sigma(c) (\pi(t) \cdot c) = \sum_{c\in C} \sigma(c) ( v(t) - \pi(t)\cdot c) \leq v(t) - \pi(t) \cdot c(t)  \ \forall \pi\in \Pi(t)
\]
Thus $c(t)$ is mixed-strategy optimal for $t$ in $C$. 
\end{proof}

\begin{corollary}
Optimal full extraction holds only if beliefs $\{ \Pi(t) \subseteq \Delta(S) : t\in T\}$ do not satisfy convex dependence. 
\end{corollary}
\begin{proof}
The proof follows the proof of Theorem 3. Suppose beliefs satisfy convex dependence, so for some $t_0\in T$,
\[
\overline{\mbox{co}} (\cup_{s\not= t_0} \Pi(s)) \subseteq \Pi(t_0)
\]
Let $v(t)> v(t_0)$ for all $t\not= t_0$. Suppose by way of contradiction that the menu  $\{ c(t)\in {\bf R}^S: t\in T\}$ achieves optimal full extraction. Then 
for $t_0$, 
\[
v(t_0) - \pi\cdot c(t_0) \geq 0 \ \forall \pi\in \Pi(t_0)
\]
while for $t\not= t_0$, 
\[
v(t) - \pi\cdot c(t_0) \leq v(t) - \pi\cdot c(t)  \ \forall \pi\in \Pi(t)
\]
and 
\[
v(t) - \pi\cdot c(t) = 0 \ \mbox{ for some } \pi\in \Pi(t)
\]
Putting these together, for $t\not= t_0$ there must be some $\pi\in \Pi(t)$ such that $v(t) - \pi\cdot c(t_0) \leq 0$. 

But for $t\not= t_0$, 
\[
\Pi(t) \subseteq \overline{\mbox{co}} (\cup_{s\not= t_0} \Pi(s)) \subseteq \Pi(t_0)
\]
Thus for $t\not= t_0$, 
\begin{eqnarray*}
v(t) - \pi\cdot c(t_0) &=& v(t) - v(t_0) + v(t_0) -\pi\cdot c(t_0)\\
&\geq & v(t) - v(t_0) > 0 \ \ \forall \pi\in \Pi(t)
\end{eqnarray*}
This is a contradiction. Thus optimal full extraction is not possible. 
\end{proof}

We close this section with some results on the limits to full extraction when convex independence fails. We identify a natural condition under which incentive compatibility limits the variation in contracts that the designer can offer, in some cases limiting the designer to a single contract. 

We start with a general definition of incentive compatibility in this setting. The menu $C= \{c(t) \in {\bf R}^S: t\in T\}$ is \emph{incentive compatible} if for each $t\in T$: 
\[
v(t) - \pi(t) \cdot c(t) \geq v(t) - \pi(t) \cdot c(s) \ \  \forall \pi\in \Pi(t), \ \forall s\not= t
\]
If $T^* \subseteq T$, say a menu $C^* = \{ c(t)\in {\bf R}^S: t\in T^*\}$ is \emph{incentive compatible for }$T^*$ if for all $t\in T^*$:
\[
v(t) - \pi(t) \cdot c(t) \geq v(t) - \pi(t) \cdot c(s) \ \  \forall \pi\in \Pi(t), \ \forall s\in T^* 
\]
Similarly, say $C^*$ is \emph{individually rational} if for all $t\in T^*$,
\[
v(t) - \pi(t) \cdot c(t) \geq 0 \ \ \forall \pi\in \Pi(t)
\] 
Note that these notions of incentive compatibility and individual rationality are consistent with the robust version implicit in the definition of full extraction, and, as in optimal full extraction, with providing agents stronger incentives in the form of optimal choices rather than maximal choices. 

Next we consider a natural condition on the richness of beliefs shared by different types. 

\begin{definition}
Beliefs $\{ \Pi(t) \subseteq \Delta(S): t\in T\}$ are \emph{fully overlapping} if for each $t, t'\in T$, $\Pi(t) \cap \Pi(t')$ has full dimension.
\end{definition}

When beliefs are fully overlapping, incentive compatibility imposes significant restrictions on possible menus of contracts, since it requires that two contracts $c(t)$ and $c(t')$ must have the same expected value according to any belief $\pi$ that is shared by types $t$ and $t'$. When the set of such shared beliefs is sufficiently rich, this observation forces $c(t)$ and $c(t')$ to be the same. This observation in turn yields the following result.

\begin{theorem}
Suppose $C = \{ c(t)\in {\bf R}^S : t\in T\}$ is incentive compatible. 
\begin{enumerate}[(i)]
  \item If $\Pi(t) \cap \Pi(t')$ has full dimension, then $c(t) = c(t')$. 
  \item If beliefs $\{ \Pi(t) \subseteq \Delta(S): t\in T\}$ are fully overlapping then $c(t) = c(t')$ for all $t, t'\in T$. 
  \item If $\Pi(t) \cap \Pi(t')$ has full dimension for each $t, t'\in T^* \subseteq T$ and $C^* = \{ c(t) \in {\bf R}^S: t\in T^*\}$ is incentive compatible for $T^*$, then $c(t) = c(t')$ for all $t, t'\in T^*$. 
\end{enumerate}
\end{theorem}
\begin{proof}
For (i), fix $t, t'\in T$ with $t\not= t'$. Since $C$ is incentive compatible, 
\[
v(t) - \pi \cdot c(t') \leq v(t) - \pi\cdot c(t) \ \ \forall \pi\in \Pi(t)
\]
and 
\[
v(t') - \pi \cdot c(t) \leq v(t') - \pi\cdot c(t') \ \ \forall \pi\in \Pi(t')
\]
Thus
\[
\pi\cdot c(t') \geq \pi\cdot c(t) \ \ \forall \pi\in \Pi(t)
\]
and
\[
\pi\cdot c(t) \geq \pi\cdot c(t') \ \ \forall \pi\in \Pi(t')
\]
Putting these together, 
\[
\pi\cdot (c(t) - c(t')) = 0 \ \ \forall \pi\in \Pi(t)\cap \Pi(t')
\]
Since $\Pi(t)\cap \Pi(t')$ has full dimension, this implies $c(t) = c(t')$.

Then (ii) follows from (i), since if beliefs are fully overlapping then $\Pi(t) \cap \Pi(t')$ has full dimension for any $t, t'\in T$, which implies that $c(t) = c(t')$ for all $t, t'\in T$. Similarly, (iii) follows from the argument used to prove (i). 
\end{proof}

Note that (ii) above also holds under a weaker condition on beliefs, provided only that there is some indexing $T = \{ t_1, \ldots ,t_T\}$ such that $\Pi(t_i) \cap \Pi(t_{i+1})$ has full dimension for each $i=1,\ldots , T-1$; similarly, (iii) holds under an analogous weakening applied to the subset $T^* \subseteq T$. 

Thus even if the designer chooses to offer a smaller menu of contracts, for example by forgoing participation for some types with lower values to focus on extracting more surplus from types with higher values, these results imply that the designer must always offer the same contract to any pair of types sharing a set of beliefs of full dimension, and when beliefs are fully  overlapping the designer can offer at most one contract. With some additional information about the designer's beliefs $\pi(d) \in \Delta(S)$, these results lead to sharp predictions regarding the designer's optimal menu of contracts. Although these results require additional restrictions that we do not focus on otherwise, they might have some independent interest. Thus we record these results next. 

\begin{theorem}
Suppose beliefs $\{ \Pi(t) \subseteq \Delta(S): t\in T\}$ are fully overlapping. Let  $T^* \subseteq T$, and let $t_1\in T^*$ satisfy $v(t_1) = \min_{t\in T^*} v(t)$. If $\pi(d) \in \Pi(t_1)$, then the deterministic contract $v(t_1)$ maximizes the designer's expected revenue among all  menus $C^* = \{ c(t) \in {\bf R}^S: t\in T^*\}$ that are incentive compatible and individually rational for $T^*$. 
\end{theorem}
\begin{proof}
Let $C^* = \{ c(t) \in {\bf R}^S: t\in T^*\}$ be incentive compatible and individually rational for $T^* \subseteq T$. By part (iii) of Theorem 4, $c(t) = c(t')= c(t_1)$ for all $t, t'\in T^*$. Since $C^*$ is individually rational, 
\[
v(t_1) - \pi \cdot c(t_1) \geq 0 \ \ \ \forall \pi\in \Pi(t_1)
\]
Thus 
\[
\pi\cdot c(t_1) \leq v(t_1) \ \ \ \forall \pi\in \Pi(t_1)
\]
If $\pi(d) \in \Pi(t_1)$, this implies $\pi(d) \cdot c(t_1) \leq v(t_1)$. Now note that the menu in which $c(t) = v(t_1)$ for each $t\in T^*$ is incentive compatible and individually rational.   Thus the deterministic contract $v(t_1)$ maximizes the designer's expected revenue among all such menus. 
\end{proof}

Note that similar conclusions follow if the designer has a set of beliefs $\Pi(d) \subseteq \Delta(S)$. In that case, the deterministic contract $v(t_1)$ is optimal for the designer among all menus that are incentive compatible and individually rational for $T^*$, provided $\Pi(d)\subseteq \Pi(t_1)$. If instead the designer uses a maxmin criterion, then the deterministic contract $v(t_1)$ maximizes the designer's minimum expected revenue among all menus that are incentive compatible and individually rational for $T^*$ as long as $\Pi(d) \cap \Pi(t_1) \not= \emptyset$. 

\section{Genericity of Full Extraction}

In this section we investigate the robustness of the necessary and sufficient conditions for full extraction we identified in the previous section. The work of Cr{\' e}mer and McLean (1988) and others uncovered the connection between correlated beliefs and full extraction in standard Bayesian mechanism design. Perhaps the most powerful and negative aspect of this
work was showing that these conditions are generic in an appropriate sense. Related work
showed that these generic conditions lead to full extraction in a wide array of settings with
private information. 

We seek a similar measure of the extent of full rent extraction and
the existence of information rents for robust notions of incentive compatibility.
To formalize this discussion, suppose $\vert S\vert \geq \vert T \vert$. Recall that in the standard Bayesian
setting, each type $t\in T$ is associated with a unique conditional distribution $\pi(t)\in \Delta(S)$,
and when $T$ and $S$ are finite, full  extraction is possible if and only if the collection $\{ \pi(t) : t\in T\}$ satisfies convex independence. In this case, it is straightforward to see that the subset of $\Delta(S)^T$ on which this condition is satisfied is an 
open set of full Lebesgue measure. 

In our setting, each type $t$ is associated with a set $\Pi(t)$ drawn from 
\[
{\cal B} = \{ \Pi\subseteq \Delta(S): \Pi \mbox{ is closed and convex}\}
\]
and full extraction is characterized in terms of conditions on the collection $\{ \Pi(t) : t\in T\} \subseteq {\cal B}^T$. To gauge how widespread the absence of information rents is in this setting, we seek to measure the size of the subset of ${\cal B}^T$ on which the various conditions characterizing full extraction and weak full extraction hold.

To make this precise, we endow ${\cal B}$  with the Hausdorff topology, and ${\cal B}^T$ with the product
topology. Let ${\cal W}$ denote the subset of ${\cal B}^T$ satisfying weak convex independence, ${\cal I}$ denote the subset of ${\cal B}^T$
satisfying convex independence, and ${\cal D}$ denote the subset of ${\cal B}^T$ satisfying convex dependence. Thus ${\cal W}$ is a set on
which weak full extraction is always possible. Similarly, ${\cal I}$ is a set of beliefs for which
 full extraction is always possible, and ${\cal D}$ is a set for which  full extraction
is never possible.

The set ${\cal B}^T$ is infinite-dimensional, even when $S$ and $T$ are finite. Thus the issue of
measuring the sizes of these sets is not straightforward due to the absence of a natural
analogue of Lebesgue measure in infinite-dimensional spaces. Genericity in these cases
is typically defined either using topological notions, such as open and dense or residual,
or using measure-theoretic notions such as prevalence. Prevalence and its complement, shyness, developed by Christensen (1974) and Hunt et al. (1992), and made relative by Anderson and Zame (2001), are analogues of Lebesgue measure 
and full Lebesgue measure that more closely mimic properties of Lebesgue measure in many
problems.\footnote{Well-known problems with interpreting topological notions of genericity are illustrated by simple
examples of open and dense sets in ${\bf R}^n$ having arbitrarily small Lebesgue measure, and residual sets of
Lebesgue measure 0.} We first give formal definitions, and then discuss some important properties
shared by these notions of genericity.

Because we are interested in the relative size of subsets of ${\cal B}^T$, we use the relative notions
of prevalence and shyness developed by Anderson and Zame (2001) for use in a convex subset which may be a shy
subset of the ambient space. The formal definitions are given below.

\begin{definition}
Let $Z$ be a topological vector space and let $C\subseteq Z$ be a convex Borel subset
of Z which is completely metrizable in the relative topology. Let $c\in C$. A universally
measurable subset $E\subseteq Z$ is \emph{shy in} $C$ \emph{at} $c$ if for each $\delta >0$ and each neighborhood $W$ of 0
in $Z$, there is a regular Borel probability measure $\mu$ on $Z$ with compact support such that
$\mbox{supp } \mu \subseteq (\delta(C-c) + c)\cap (W+c)$ and $\mu(E+z) =0$ for every $z\in Z$.\footnote{A set $E\subseteq Y$ is universally measurable if for every Borel measure $\eta$ on $Y$, $E$ belongs to the completion
with respect to $\eta$ of the sigma algebra of Borel sets.} The set $E$ is \emph{shy
in} $C$ if it is shy at each point $c\in C$. A (not necessarily universally measurable) subset $F\subseteq C$ is \emph{shy in} $C$ if it is contained in a shy universally measurable set. A subset $K\subseteq C$ is \emph{prevalent in} $C$ if its complement $C\setminus K$ is shy in $C$.
\end{definition}

Like Lebesgue measure 0, relative shyness and prevalence have many properties desirable for measure-theoretic notions of ``smallness'' and ``largeness'': relative shyness is translation invariant, preserved under countable unions, and coincides with Lebesgue measure 0 in ${\bf R}^n$, and no relatively open set is relatively shy.

We note a simple but important property common to both residual and relative prevalence as notions of genericity with respect to subsets of ${\cal B}^T$.

\begin{lemma}
Let $X\subseteq {\cal B}^T$ be universally measurable. If $X^c = {\cal B}^T \setminus X$ has a non-empty
relative interior, then $X$ is neither residual nor relatively prevalent in ${\cal B}^T$.
\end{lemma}
\begin{proof}
For relative prevalence the result is immediate from the definitions and the fact that
no relatively open set is relatively shy. To see that $X$ is not residual in ${\cal B}^T$, note that ${\cal B}^T$ 
is a compact metric space, hence a Baire space. The conclusion then follows immediately from the Baire Category Theorem.
\end{proof}

From this simple observation, we conclude that full extraction is neither generically possible nor generically impossible. Similarly, optimal full extraction is neither generically possible nor generically impossible. 

\begin{theorem}
Let $\vert S\vert \geq \vert T \vert$. Neither ${\cal I}$ nor ${\cal D}$ is residual in ${\cal B}^T$. Neither ${\cal I}$ nor ${\cal D}$ is
relatively prevalent in ${\cal B}^T$.
\end{theorem}
\begin{proof}
Both ${\cal I}$ and ${\cal D}$ are Borel sets, hence are universally measurable. By definition,
${\cal I}\cap {\cal D} = \emptyset$, so ${\cal I}\subseteq {\cal D}^c$ and ${\cal D}\subseteq {\cal I}^c$. The results will all follow provided both ${\cal I}$ and ${\cal D}$ have non-empty relative interior. In both cases, we will establish this by constructing relative interior points.

First consider ${\cal I}$. Choose $\{ \pi(t)\in \Delta(S): t\in T\}$ such that $\pi(t) \not\in \overline{\mbox{co}}\{ \pi(t'): t'\not= t\}$  for
each $t$, that is, such that $\{ \{\pi(t)\} : t\in T\}$ satisfies convex independence. Choose $\varepsilon > 0$ such that
\[
\forall t\in T: \ \ B_\varepsilon(\pi(t)) \cap \overline{\mbox{co}} \left( \cup_{t'\not= t}B_\varepsilon(\pi(t'))\right) =\emptyset
\]
that is, such that $\{ B_\varepsilon(\pi(t)): t\in T\}$ also satisfies convex independence.\footnote{As in section 3, here for $\pi\in \Delta(S)$ and $\beta>0$, $B_\beta(\pi)$ denotes the ball of radius $\beta$ about $\pi$ in $\Delta(S)$, so $B_\beta(\pi) = \{ \pi'\in \Delta(S): \Vert \pi'-\pi \Vert < \beta\}$.} Now if $\Pi\in {\cal B}$ and $d(\Pi, \{ \pi(t)\})<\varepsilon$, then $\Pi \subseteq B_\varepsilon(\pi(t))$.\footnote{Here $d(A,B)$ denotes the Hausdorff distance between $A, B\in {\cal B}$, defined by
\[
d(A,B) = \max\left\lbrace\sup_{x\in A} \mbox{dist}(x, B), \sup_{y\in B} \mbox{dist}(y,A)\right\rbrace
\]} Thus any collection $\{ \Pi(t)\in {\cal B}: t\in T\}$ such that $d(\Pi(t), \{ \pi(t)\})<\varepsilon$ for each $t$ must satisfy convex independence as well. From this we conclude that $\{ \{\pi(t)\} : t\in T\}$ is a relative interior point of ${\cal I}$. 

Next, consider ${\cal D}$. Fix $t_0\in T$. Choose $\Pi(t_0)\in {\cal B}$ such that $\mbox{rint}\Pi(t_0) \not= \emptyset$. Fix $\varepsilon >0$ and choose $\bar \pi(t_0)\in \mbox{rint}\Pi(t_0)$ such that $B_\varepsilon(\bar \pi(t_0)) \subseteq \Pi(t_0)$. Now choose $\{ \Pi(t) \subseteq \Delta(S): t\in T\setminus \{ t_0\}\}$ such that 
\[
\overline{\mbox{co}} ( \cup_{t\not= t_0} \Pi(t)) \subseteq B_{\varepsilon/4}(\bar \pi(t_0))
\]
In particular then, $\{ \Pi(t): t\in T\}$ satisfies convex dependence, so the collection $\{ \Pi(t): t\in T\} \in {\cal D}$. 

If $\{ \bar \Pi(t) \in {\cal B}: t\in T\setminus \{ t_0\}\}$ is any collection such that $d(\bar \Pi(t), \Pi(t)) < \varepsilon/4$ for each $t\not= t_0$, then
\[
\overline{\mbox{co}} ( \cup_{t\not= t_0}\bar \Pi(t)) \subseteq B_{\varepsilon/2}(\bar \pi(t_0))
\]
Finally, if $\Pi\in {\cal B}$ and $d(\Pi, \Pi(t_0))<\varepsilon/2$, then $B_{\varepsilon/2}(\bar \pi(t_0))\subseteq \Pi$. Putting these observations together, any collection $\{ \bar \Pi(t)\in {\cal B}: t\in T\}$ such that $d(\bar \Pi(t_0), \Pi(t_0))<\varepsilon/2$ and $d(\bar \Pi(t), \Pi(t))<\varepsilon/4$ for each $t\not= t_0$ will also satisfy convex dependence, and hence belongs to ${\cal D}$.
\end{proof}

\section{Discussion and Extensions}

In this section we consider several extensions of our main results, again motivated by questions of robustness. We first explore the performance of given extraction mechanisms when beliefs are perturbed. This allows us to study the extent to which a designer can construct a menu of contracts that performs well even if it is based on beliefs that are incorrect.  Then we show that several of our main results carry over to other standard models of ambiguity, including versions of maxmin and alpha-maxmin expected utility. 

One criticism of the standard full extraction results of Cr{\' e}mer and McLean (1985, 1988) and McAfee and Reny (1992) is that they rely on the designer knowing agents' beliefs precisely. As such, the information required by the designer is substantial, and the conclusion that full extraction is possible whenever convex independence holds is perhaps fragile. A main motivation for our original questions and results was the possibility of misspecification of agents' beliefs in the standard surplus extraction problem. Our main results can thus already be interpreted as providing conditions under which full extraction might be robust to misspecification of beliefs, and illustrating limits on extraction as a consequence. As we detailed throughout the paper, our main results can also be derived from agents' perceptions of uncertainty. In this interpretation, our results are subject to criticisms analogous to the standard case, as the designer must know agents' (sets of) beliefs precisely for our full extraction results as well. For either interpretation, it is thus natural to ask how robust a given extraction mechanism might be to mistakes in the designer's specification of agents' beliefs. To explore this question, we first consider a more general notion of surplus extraction.\footnote{Parallel results to Theorems 7 and 8 below can be derived for an analogous more general notion of optimal full extraction, along the lines of Corollary 2.}      

\begin{definition}
Let $v:T\to {\bf R}$ be given, and let $b\in {\bf R}_+$. The menu $C= \{ c(t) \in {\bf R}^S : t\in T\}$ \emph{leaves at most} $b$ \emph{surplus for beliefs} $\{ \Pi(t)\subseteq \Delta(S): t\in T\}$ if for each $t\in T$:
\begin{eqnarray*}
v(t) -\pi \cdot c(t) &\geq & 0 \ \ \forall \pi\in \Pi(t)\\
v(t) - \pi\cdot c(s) &\leq& 0 \ \ \forall \pi\in \Pi(t), \ \forall s\not= t
\end{eqnarray*}
and
\[
\min_{\pi\in \Pi(t)} \left( v(t) - \pi\cdot c(t) \right) \leq b
\]
\end{definition}
Note that when $b=0$, this notion coincides with full extraction.

We start with a simple observation. When misspecification enlarges sets of beliefs, any full extraction menu designed for these larger sets will still extract surplus under agents' true beliefs, leaving  at most an amount of expected surplus bounded by the amount of misspecification. To formalize this observation, suppose beliefs $\{ \Pi(t) \subseteq \Delta(S) : t\in T\}$ satisfy convex independence. Let $v:T\to {\bf R}$ be given and suppose the menu $C=\{ c(t) \in {\bf R}^S: t\in T\}$ achieves full extraction given $v$ and beliefs $\{ \Pi(t): t\in T\}$. Now suppose the designer has misspecified agents' beliefs, and that  true beliefs are a subset of these, so  $\Pi'(t) \subseteq \Pi(t)$ for each $t\in T$. The menu $C$ might fail to extract all of the surplus in this case, but it will still be individually rational and incentive compatible. To see this, note that for each $t\in T$:
\begin{eqnarray*}
v(t) - \pi\cdot c(t) & \geq & 0 \ \ \forall \pi\in \Pi'(t)\\
v(t) - \pi\cdot c(s) & \leq & 0 \ \ \forall \pi\in \Pi'(t), \ \forall s\not= t
\end{eqnarray*}
since $\Pi'(t) \subseteq \Pi(t)$ for each $t$, and the menu $C$ achieves full extraction for the beliefs $\{ \Pi(t): t\in T\}$. Full extraction might not hold for the true beliefs $\{ \Pi'(t): t\in T\}$, however: although $v(t) -\pi\cdot c(t) = 0$ for some $\pi\in \Pi(t)$, expected surplus $v(t) -\pi\cdot c(t)$ might be strictly positive for all $\pi\in \Pi'(t)$. 
Although $C$ might not achieve full extraction, we can nonetheless give an upper bound on the surplus $C$ leaves to agents, based on the Hausdorff distance between the true beliefs and the beliefs used to construct the menu $C$. To that end, fix $t\in T$. Because $C$ achieves full extraction for beliefs $\{ \Pi(t): t\in T\}$,  there exists $\pi(t)\in \Pi(t)$ such that $v(t) - \pi(t) \cdot c(t)=0$.  Then for any $\pi \in \Pi'(t)$, 
\begin{eqnarray*}
v(t) - \pi \cdot c(t) &=& \vert \left( v(t) - \pi \cdot c(t) \right) - \left( v(t) - \pi(t) \cdot c(t) \right) \vert\\
&=& \vert (\pi -\pi(t) ) \cdot c(t)\vert \leq \Vert \pi - \pi(t) \Vert \Vert c(t) \Vert
\end{eqnarray*}
 Thus   
\[
\min_{\pi\in \Pi'(t)} \left( v(t) - \pi\cdot c(t)\right) \leq \min_{\pi\in \Pi'(t)} \Vert \pi - \pi(t) \Vert \Vert c(t) \Vert\\
\leq d(\Pi(t), \Pi'(t))  \Vert c(t) \Vert  
\]
Repeating this argument for each $t\in T$ shows that the menu $C$ leaves at most $d\Vert C \Vert$ surplus for beliefs $\{ \Pi'(t): t\in T\}$, where $\Vert C \Vert = \max\limits_{t\in T} \Vert c(t) \Vert$ and $d = \max\limits_{t\in T} d(\Pi(t), \Pi'(t))$.  The surplus $d\Vert C\Vert$ can be viewed as the cost of misspecification, or the cost of robustness concerns for the designer in this case. 

We record this observation below. 

\begin{theorem}
Suppose beliefs $\{ \Pi(t) \subseteq \Delta(S) : t\in T\}$ satisfy convex independence. Let $v:T\to {\bf R}$ be given and suppose the menu $C =\{ c(t) \in {\bf R}^S: t\in T\}$ achieves full extraction given $v$ and beliefs $\{ \Pi(t): t\in T\}$. Then $C$ leaves at most $d\Vert C \Vert$ surplus for any beliefs $\{ \Pi'(t) \subseteq \Delta(S): t\in T\}$ with $\Pi'(t) \subseteq \Pi(t)$ for each $t\in T$, where $\Vert C \Vert = \max\limits_{t\in T} \Vert c(t) \Vert$ and $d = \max\limits_{t\in T} d(\Pi(t), \Pi'(t))$. 
\end{theorem}

The leading example of $\varepsilon$-contamination illustrates this result, and provides a particularly natural setting in which misspecification might correspond to enlarging agents' sets of beliefs. Here suppose $\{ \pi(t) \in \Delta(S): t\in T\}$ is given and $\varepsilon >0$. Then  note that for any $\varepsilon ' \geq 0$,
\[
d(\Pi_\varepsilon(t), \Pi_{\varepsilon'}(t)) \leq 2\vert \varepsilon - \varepsilon'\vert
\]
To see this, let $\pi_\varepsilon \in \Pi_\varepsilon(t)$, so $\pi_\varepsilon = (1-\varepsilon)\pi(t)  + \varepsilon \pi$ for some $\pi\in \Delta(S)$. Then set $\pi_{\varepsilon'} = (1-\varepsilon') \pi(t) + \varepsilon' \pi$. By definition, $\pi_{\varepsilon'} \in \Pi_{\varepsilon'}(t)$, and
\[
\pi_\varepsilon - \pi_{\varepsilon'} = (\varepsilon' - \varepsilon) \pi(t) + (\varepsilon - \varepsilon') \pi = (\varepsilon - \varepsilon') (\pi - \pi(t))
\]
so 
\[
\Vert \pi_\varepsilon - \pi_{\varepsilon'} \Vert \leq \vert \varepsilon - \varepsilon'\vert \Vert \pi -  \pi(t)\Vert \leq 2 \vert \varepsilon - \varepsilon'\vert
\]
Thus $d(\Pi_\varepsilon(t) , \Pi_{\varepsilon'}(t)) \leq 2\vert \varepsilon - \varepsilon'\vert$. 

Putting this observation together with the previous ones yields the following corollary.  

\begin{corollary}
Suppose $\{  \pi(t) \in \Delta(S): t\in T\}$ satisfies convex independence. Let $\varepsilon >0$ such that $\{\Pi_\varepsilon(t): t\in T\}$ satisfies convex independence, and let $v:T\to {\bf R}$ be given. Suppose the menu $C = \{ c(t) \in {\bf R}^S :t\in T\}$ achieves full extraction given $v$ and beliefs    $\{\Pi_\varepsilon(t): t\in T\}$. For any $0\leq \varepsilon ' < \varepsilon$, the menu $C$ leaves at most $2(\varepsilon - \varepsilon') \Vert C \Vert$ surplus for beliefs $\{\Pi_{\varepsilon'}(t): t\in T\}$. In particular, $C$ leaves at most $2\varepsilon \Vert C \Vert $ surplus for beliefs $\{ \pi(t) \in \Delta(S): t\in T\}$. 
\end{corollary}

More generally, misspecification might not simply enlarge the set of beliefs for each type. In such cases a menu designed for full extraction with respect to one set of beliefs might fail to be individually rational or incentive compatible, as well as failing full extraction, for agents' true beliefs, even if the misspecification is small.\footnote{This issue also arises in the standard setting with a unique belief for each type.} We show next that a virtual extraction result holds for mechanisms that are robust to small perturbations in beliefs. Thus a designer concerned about such misspecification could choose a menu that remains individually rational and incentive compatible for all sufficiently small perturbations in beliefs in exchange for extracting all but $\varepsilon$ surplus, for any $\varepsilon >0$.\footnote{This is similar in spirit to results in Chen and Xiong (2013) in the standard setting. Chen and Xiong (2013) show that extraction mechanisms can be chosen to be robust to sufficiently small changes in priors in the universal type space, in that for each $\varepsilon >0$, if a mechanism from a particular class extracts all but at most $\varepsilon$ surplus for a given prior, then the same mechanism also extracts all but $\varepsilon$ surplus for all priors in a sufficiently small weak-$^*$ neighborhood of the original prior.}

\begin{theorem}
Suppose beliefs $\{ \Pi(t) \subseteq \Delta(S): t\in T\}$ satisfy convex independence and $v:T\to {\bf R}$ is given. For each $\varepsilon >0$ there exists a menu $C=\{ c(t) \in {\bf R}^S : t\in T\}$ and there exists $\delta >0$ such that  $C$ leaves at most $\varepsilon$ surplus for   for any beliefs $\{ \Pi'(t) \subseteq \Delta(S): t\in T\}$ with $d(\Pi(t), \Pi'(t)) < \delta$ for all $t\in T$.
\end{theorem} 

\begin{proof}
Let $v:T\to {\bf R}$ be given and fix $\varepsilon >0$. Fix $t\in T$. Since $\{ \Pi(t): t\in T\}$ satisfies convex independence, there exists $z(t)\in {\bf R}^S$ such that 
\[
\pi\cdot z(t) < 0 \ \ \forall \pi\in \Pi(t)
\]
and
\[
\pi\cdot z(t) > 0 \ \ \forall \pi\in \Pi(s), \ \forall s\not= t
\]
Then there exists $ \delta' >0$ such that for any beliefs $\{ \Pi'(s): s\in T\}$ with $d(\Pi'(s), \Pi(s)) <  \delta'$ for all $s$, then 
\[
\pi\cdot z(t) < 0 \ \ \forall \pi\in \Pi'(t)
\]
and
\[
\pi\cdot z(t) > 0 \ \ \forall \pi\in \Pi'(s), \ \forall s\not= t
\]
Now set $\alpha(t)>0$ sufficiently large so that
\[
\alpha(t) > \max_{\substack{\pi\in \Pi(s)\\ s\not= t}} \frac{v(s)-v(t) +\varepsilon}{\pi\cdot z(t)}
\]
Since $\pi\cdot z(t)>0$ for all $\pi\in \Pi(s)$ and all $s\not= t$, and $T$ is finite, such an $\alpha(t)>0$ exists.  
Set 
\[
c(t) = v(t) -\frac{\varepsilon}{2} + \left\vert \max_{\pi\in \Pi(t)} \alpha(t) \left( \pi\cdot z(t) \right)\right\vert + \alpha(t) z(t)
\]
Then for type $t$, 
\begin{eqnarray*}
v(t) - \pi\cdot c(t) &\geq& \frac{\varepsilon}{2} \ \ \forall \pi\in \Pi(t)\\
v(t) -\pi\cdot c(t) &=& \frac{\varepsilon}{2} \ \ \mbox{ for some } \pi\in \Pi(t)
\end{eqnarray*}
while for types $s\not= t$,
\[
v(s) - \pi\cdot c(t) < -\frac{\varepsilon}{2} \ \ \forall \pi\in \Pi(s)
\]
Then there exists $\delta >0$ with $\delta \leq \delta'$ such that if beliefs $\{ \Pi'(s): s\in T\}$ satisfy $d(\Pi'(s), \Pi(s)) <\delta$ for all $s\in T$, then for type $t$:
\begin{eqnarray*}
v(t) - \pi\cdot c(t) &\geq & 0 \ \ \forall \pi\in \Pi'(t)\\
\min_{\pi\in \Pi'(t)} \left( v(t) -\pi\cdot c(t)\right) &\leq& \varepsilon
\end{eqnarray*}
and for types $s\not= t$:
\[
v(s) - \pi\cdot c(t) <0 \ \ \forall \pi\in \Pi'(s)
\]
Repeating this argument for all $t\in T$ yields a menu $C= \{c(t): t\in T\}$ that leaves at most $\varepsilon$ surplus for any  beliefs $\{\Pi'(t) : t\in T\}$ with $d(\Pi'(t), \Pi(t)) < \delta$ for all $t\in T$. 
\end{proof}

Again returning to the leading example to illustrate yields the following corollary.

\begin{corollary}
Suppose $\{  \pi(t)\in \Delta(S): t\in T\}$ satisfies convex independence. Let $\varepsilon >0$ such that $\{\Pi_\varepsilon(t): t\in T\}$ satisfies convex independence, and let $v:T\to {\bf R}$ be given. For each $\eta >0$ there exists a menu $C = \{ c(t)\in {\bf R}^S:t\in T\}$ and there exists $\delta >0$ such that $C$ leaves at most $\eta$ surplus for any beliefs $\{\Pi_{\varepsilon'}(t): t\in T\}$ with $\vert \varepsilon' - \varepsilon\vert < \delta$. 
\end{corollary}

Next we consider the interpretation of our results as robustness to agents' perceptions of ambiguity, and show that some of our main results can be extended to other common models of ambiguity. For this extension, as in our basic model, we assume that each type $t\in T$ has some information rent $v(t)\in {\bf R}$, and evaluates stochastic contracts using a binary relation $\succsim_t$ on ${\bf R}^S$. Here we assume that each binary relation $\succsim_t$ is represented by a function $V_t:{\bf R}^S \to {\bf R}$, thus that $\succsim_t$ is complete for each $t\in T$. We focus on a simple class of utility functions that can be related to expected surplus with respect to some set of beliefs, as defined below.

\begin{definition}
Utility $V_t:{\bf R}^S\to {\bf R}$ is \emph{belief-based} if there exists a closed, convex set $\Pi(t) \subseteq \Delta(S)$ such that for each $v(t) \in {\bf R}$ and for all $x\in {\bf R}^S$, 
\[
V_t(x) = \pi\cdot (v(t) - x) = v(t) - \pi\cdot x \ \ \mbox{ for some } \pi\in \Pi(t)
\] 
and for all $r\in {\bf R}$, 
\[
V_t(r+ x) = - r + V_t(x)
\]
If $V_t$ is belief-based, we take $\Pi(t) \subseteq \Delta(S)$ to be the smallest such set, with respect to set inclusion, and we  say  $\Pi(t)$ are \emph{beliefs for} $t$ in this case. 
\end{definition}

Versions of standard maxmin expected utility and alpha-maxmin expected utility give two examples of belief-based utilities. For these  examples, let $\Pi(t) \subseteq \Delta(S)$ be closed and convex. Then $V_t:{\bf R}^S\to {\bf R}$ given by
\[
V_t(x) = \min_{\pi\in \Pi(t)} \pi \cdot (v(t) - x) = \min_{\pi\in \Pi(t)} (v(t) - \pi\cdot x)
\]
is belief-based. Similarly, if $\alpha \in [0,1]$, then $V_t:{\bf R}^S\to {\bf R}$ given by
\[
V_t(x) = \alpha \min_{\pi\in \Pi(t)} \pi \cdot (v(t) - x) + (1-\alpha) \max_{\pi\in \Pi(t)} \pi \cdot (v(t) - x) 
\]
is also belief-based. In each case, the set $\Pi(t)$ gives beliefs for $t$. 

For belief-based utilities, full extraction is possible whenever beliefs satisfy convex independence, as in our basic model. To show this, we start by defining full extraction in this setting. 

\begin{definition}
\emph{Full extraction holds for belief-based utilities} $\{ V_t: t\in T\}$ if for any $v:T\to {\bf R}$ there exists a menu $\{ c(t) \in {\bf R}^S : t\in T\}$ such that  for each $t\in T$:
\begin{eqnarray*}
V_t(c(t)) &=& 0\\
V_t(c(s)) &\leq& 0 \ \ \forall s\not= t
\end{eqnarray*}
\end{definition}

As the following theorem shows, full extraction holds for belief-based utilities $\{ V_t: t\in T\}$ whenever corresponding beliefs $\{ \Pi(t): t\in T\}$ satisfy convex independence. 

\begin{theorem}
Suppose for each $t\in T$,  $V_t$ is belief-based with corresponding beliefs $\Pi(t)$. If  $\{ \Pi(t): t\in T\}$ satisfies convex independence, then full extraction holds for belief-based utilities $\{ V_t: t\in T\}$. 
\end{theorem}

\begin{proof}
Let $v:T\to {\bf R}$ be given. Fix $t\in T$. By convex independence, there exists $z(t) \in  {\bf R}^S$ such that 
\[
\pi\cdot z(t) \leq 0 \ \ \forall \pi\in \Pi(t)
\]
and
\[
\pi\cdot z(t) >  0 \ \ \forall \pi\in \Pi(s), \ \ \forall s\not= t
\]
Set $c_0(t) = v(t) + \alpha(t) z(t)$ where $\alpha(t) \in {\bf R}_+$. Choose $\alpha(t) >0$ such that 
\[
\alpha(t) > \max_{\substack{\pi\in \Pi(s)\\ s\not= t }} \frac{v(s) - v(t)}{\pi\cdot z(t)}
\]
Since $\pi\cdot z(t)>0$ for all $\pi\in \Pi(s)$ and all $s\not= t$, and $T$ is finite, such an $\alpha(t)>0$ exists.  
Then set $p(t) = V_t(c_0(t))$ and set 
\[
c(t) = v(t) + p(t) + \alpha(t) z(t) = p(t) + c_0(t)
\]
Note that $p(t) = V_t(c_0(t)) \geq 0$, as $V_t(c_0(t)) = v(t) - v(t) -\alpha(t) \left(\pi\cdot z(t)\right)$ for some $\pi\in \Pi(t)$, where $\pi\cdot z(t)\leq 0$ for all $\pi \in \Pi(t)$ and $\alpha(t) >0$. Then for type $t$,
\[
V_t(c(t)) = -p(t) + V_t( c_0(t)) = -V_t(c_0(t)) + V_t(c_0(t)) = 0 
\]
and for types $s\not= t$,
\begin{eqnarray*} 
V_s(c(t)) &=& - p(t) + V_s( c_0(t))\\
&=& - p(t) + v(s) - v(t)  - \alpha(t) \left( \pi\cdot z(t)\right) \ \ \mbox{ for some } \pi\in \Pi(s)\\
&\leq& v(s) - v(t) - \alpha(t) \left( \pi\cdot z(t)\right) \\
&<& 0 \ \ \mbox{ by choice of } \alpha(t)
\end{eqnarray*}
Repeating this argument for each $t\in T$ yields a menu $\{ c(t) : t\in T\}$ that achieves full extraction. 
\end{proof}

Finally, we focus on the central case of maxmin expected surplus. In this case, full extraction again holds only if the corresponding beliefs $\{ \Pi(t): t\in T\}$ do not satisfy convex dependence. 

\begin{theorem}
Suppose for each $t\in T$, $V_t(x) = \min\limits_{\pi\in \Pi(t)} \pi\cdot (v(t) - x)$ for some closed, convex set $\Pi(t) \subseteq \Delta(S)$. Full extraction holds only if beliefs $\{ \Pi(t): t\in T\}$ do not satisfy convex dependence. 
\end{theorem}

\begin{proof}
Suppose beliefs satisfy convex dependence, and suppose 
\[
\overline{\mbox{co}} \left( \cup_{t\not= t_0} \Pi(t)\right) \subseteq \Pi(t_0)
\]
Let $v:T\to {\bf R}$ such that $v(t) > v(t_0)$ for $t\not= t_0$. Suppose by way of contradiction that full extraction holds. Thus there exists a menu $\{ c(t): t\in T\}$ such that for each $t\in T$:
\begin{eqnarray*}
V_t(c(t)) &=& 0\\
V_t(c(s)) &\leq& 0 \ \ \forall s\not= t  
\end{eqnarray*}
Then fix $t\not= t_0$. Without loss of generality write $c(t_0) = v(t_0) + z(t_0)$ for $z(t_0)\in {\bf R}^S$. Then 
\begin{eqnarray*}
0 &=& V_{t_0}(c(t_0)) = v(t_0) - v(t_0) -\max_{\pi\in \Pi(t_0)} \pi\cdot z(t_0)\\
\Rightarrow & \ & \max_{\pi\in \Pi(t_0)} \pi\cdot z(t_0) = 0
\end{eqnarray*}
Then for $t\not= t_0$, 
\begin{eqnarray*}
V_t(c(t_0)) &=& v(t) - \max_{\pi\in \Pi(t)} \pi\cdot (v(t_0) + z(t_0))\\
&=& v(t) - v(t_0) - \max_{\pi\in \Pi(t)} \pi\cdot z(t_0)
\end{eqnarray*}
But $\Pi(t) \subseteq \Pi(t_0)$, so 
\[
\max_{\pi\in \Pi(t)} \pi\cdot z(t_0) \leq  \max_{\pi\in \Pi(t_0)} \pi \cdot z(t_0) = 0
\]
So 
\[
V_t(c(t_0)) \geq v(t) - v(t_0) >0
\]
This is a contradiction. Thus full extraction is not possible. 
\end{proof}

As an immediate corollary of Theorem 6, Theorems 9 and 10 imply that for maxmin expected surplus, whenever $\vert S\vert \geq \vert T \vert$ full extraction is neither generically possible nor generically impossible.


\begin{thebibliography}{99}


\bibitem{ } Albert, M., Conitzer, V., Lopomo, G., and P. Stone (2019): ``Mechanism Design for Correlated Valuations: Efficient Methods for Revenue Maximization,'' working paper. 

\bibitem{} Ahn, D. (2007): ``Hierarchies of Ambiguous Beliefs,'' \emph{Journal of Economic Theory}, 136, 286-301.

\bibitem{} Anderson, R. M., and W. R. Zame (2001): ``Genericity with Infinitely Many Parameters,''
\emph{Advances in Theoretical Economics}, 1, Article 1.

\bibitem{} Aumann, R. J. (1962): ``Utility Theory without the Completeness Axiom,'' \emph{Econometrica}, 30,
445-462.

\bibitem{} Aumann, R. J. (1964): ``Utility Theory without the Completeness Axiom: A Correction,'' \emph{Econometrica},
32, 210-212.

\bibitem{ } Barelli, P. (2009): ``On the Genericity of Full Surplus Extraction in Mechanism Design,'' \emph{Journal of Economic Theory}, 144, 1320-1332.

\bibitem{} Bewley, T. F. (1986): ``Knightian Decision Theory: Part I,'' Discussion paper, Cowles Foundation.

\bibitem{} Bewley, T. F. (2002): ``Knightian Decision Theory: Part I,'' \emph{Decisions in Economics and Finance}, 2,
79-110.

\bibitem{} Bodoh-Creed, A. (2012): ``Ambiguous Beliefs and Mechanism Design,'' \emph{Games and Economic Behavior}, 75, 518-537.

\bibitem{} Bose, S., and A. Daripa (2009): ``A Dynamic Mechanism and Surplus Extraction under Ambiguity,'' \emph{Journal of Economic Theory}, 144, 2084-2114.


\bibitem{ } Bose, S., E. Ozdenoren, and A. Pape (2006): ``Optimal Auctions with Ambiguity,'' \emph{Theoretical
Economics}, 1, 411-438.

\bibitem{} Bose, S., and L. Renou (2014): ``Mechanism Design with Ambiguous Communication Devices,'' \emph{Econometrica}, 82, 1853-1872.

\bibitem{ } Che, Y.-K. and J. Kim (2006): ``Robustly Collusion-Proof Mechanisms,'' \emph{Econometrica}, 74, 1063-1107.

\bibitem{} Chen, Y.-C., and S. Xiong (2011): ``The Genericity of Beliefs-Determine-Preferences Models Revisited,'' \emph{Journal of Economic Theory}, 146, 751-761.

\bibitem{} Chen, Y.-C., and S. Xiong (2013): ``Genericity and Robustness of Full Surplus Extraction,'' \emph{Econometrica}, 81(2), 825-847.

\bibitem{} Chiesa, A., S. Micali, and Z. Zhu (2015): ``Knightian Analysis of the Vickrey Mechanism,'' \emph{Econometrica}, 83(5), 1727-1754.

\bibitem{ } Christensen, J. P. R. (1974): \emph{Topology and Borel Structure}. Amsterdam: North Holland.


\bibitem{ } Cr{\' e}mer, J.-J., and R. McLean (1985): ``Optimal Selling Strategies under Uncertainty for a
Discriminatory Monopolist when Demands Are Interdependent,'' \emph{Econometrica}, 53, 345-61.

\bibitem{ } Cr{\' e}mer, J.-J., and R. McLean (1988): ``Full Extraction of the Surplus in Bayesian and Dominant Strategy Auctions,''
\emph{Econometrica}, 56, 1247-57.

\bibitem{} De Tillio, A., N. Kos, and M. Messner (2016): ``The Design of Ambiguous Mechanisms,'' \emph{Review of Economic Studies},  84, 237-276.

\bibitem{ } Dubra, J., F. Maccheroni, and E. A. Ok (2004): ``Expected Utility Theory without the
Completeness Axiom,'' \emph{Journal of Economic Theory}, 115, 118-133.

\bibitem{} Fu, H., N. Haghpanah, J. Hartline, and R. Kleinberg (2017): ``Full Surplus Extraction from Samples,'' working paper. 


\bibitem{ } Ghirardato, P., F. Maccheroni, and M. Marinacci (2004): ``Differentiating Ambiguity
and Ambiguity Attitude,'' \emph{Journal of Economic Theory}, 118, 133-173.

\bibitem{ } Ghirardato, P., F. Maccheroni, M. Marinacci, and M. Siniscalchi (2003): ``A Subjective
Spin on Roulette Wheels,'' \emph{Econometrica}, 71, 1897-1908.

\bibitem{ } Gilboa, I., F. Maccheroni, M. Marinacci, and D. Schmeidler (2010): ``Objective and
Subjective Rationality in a Multiple Prior Model,'' \emph{Econometrica}, 78(2), 755-770.

\bibitem{ }  Girotto, B., and S. Holzer (2005): ``Representation of Subjective Preferences Under Ambiguity,''
\emph{Journal of Mathematical Psychology}, 49, 372-382.

 \bibitem{ }  Heifetz, A., and Z. Neeman (2006): ``On the Generic Impossibility of Full Surplus Extraction
in Mechanism Design,'' \emph{Econometrica}, 74, 213-233.

 \bibitem{ }  Hunt, B., T. Sauer, and J. Yorke (1992): ``Prevalence: A Translation Invariant `Almost
Every' on Infinite Dimensional Spaces,'' \emph{Bulletin (New Series) of the American Mathematical
Society}, 27, 217-238.

\bibitem{} Jehiel, P., M. Meyer-ter-Vehn, and B. Moldovanu (2012): ``Locally Robust Implementation and its Limits,'' \emph{Journal of Economic Theory}, 147, 2439-2452.

\bibitem{ } Laffont, J.-J., and D. Martimort (2000): ``Mechanism Design with Collusion and Correlation,'' \emph{Econometrica}, 68, 309-342.

\bibitem{ } Lopomo, G., L. Rigotti, and C. Shannon (2009): ``Uncertainty in Mechanism Design,'' working paper. 

\bibitem{ }   McAfee, P. R., and P. J. Reny (1992): ``Correlated Information and Mechanism Design,''
\emph{Econometrica}, 60, 395-421.

\bibitem{ }  Neeman, Z. (2004): ``The Relevance of Private Information in Mechanism Design,'' \emph{Journal of
Economic Theory}, 117, 55-77.

 \bibitem{ }  Ok, E. A. (2002): ``Utility Representation of an Incomplete Preference Relation,'' \emph{Journal of
Economic Theory}, 104, 429-449.

\bibitem{ } Peters, M. (2001): ``Surplus Extraction and Competition,'' \emph{Review of Economic Studies}, 68, 613-631.


 \bibitem{ }  Rigotti, L., and C. Shannon (2005): ``Uncertainty and Risk in Financial Markets,'' \emph{Econometrica},
73, 203-243.

\bibitem{ } Robert, J. (1991): ``Continuity in Auction Design,'' 
\emph{Journal of Economic Theory}, 55, 169-179.


\bibitem{ } Shapley, L. S., and M. Baucells (2008): ``Multiperson Utility,'' \emph{Games and Economic Behavior},
62, 329-347.

\bibitem{ } Wolitzky, A. (2016): ``Mechanism Design with Maxmin Agents: Theory and an Application to Bilateral Trade,'' \emph{Theoretical Economics}, 11, 971-1004.

\end{thebibliography}
\end{document}